\documentclass[12pt,onecolumn,twoside]{IEEEtran}

\usepackage{calc}

\usepackage{multirow}
\usepackage{cite}
\usepackage{graphicx,subfigure}
\usepackage{psfrag}
\usepackage{amsmath,amssymb,amsthm}
\usepackage{color}
\usepackage{tikz} 

\usepackage{bbm}

\usepackage{cases}

\DeclareMathOperator*{\defeq}{\triangleq}

\newtheorem{theorem}{Theorem}

\newtheorem{lemma}{Lemma}

\newcommand{\bit}{\begin{itemize}}
\newcommand{\eit}{\end{itemize}}

\newcommand{\bc}{\begin{center}}
\newcommand{\ec}{\end{center}}

\newcommand{\ba}{\begin{array}}
\newcommand{\ea}{\end{array}}

\newcommand{\beq}{\begin{equation}}
\newcommand{\eeq}{\end{equation}}

\newcommand{\beqn}{\begin{equation*}}
\newcommand{\eeqn}{\end{equation*}}

\newcommand{\bean}{\begin{eqnarray*}}
\newcommand{\eean}{\end{eqnarray*}}
\newcommand{\bea}{\begin{eqnarray}}
\newcommand{\eea}{\end{eqnarray}}

\def\E{\mathbb{E}}

\newcommand{\Lc}{{\mathcal L}}

\newcommand{\Rc}{{\mathcal R}}
\newcommand{\Sc}{{\mathcal S}}

\newcommand{\Zc}{{\mathcal Z}}

\newcommand{\non}{\nonumber}

\newcommand{\Hen}{\mathbb{H}}
\newcommand{\hen}{\mathrm{h}}
\newcommand{\Imu}{\mathbb{I}}

\newcommand{\bln}{n}

\newcommand{\Ho}{\mathcal{H}_{\text{out}}}
\newcommand{\Hob}{\bar{\mathcal{H}}_{\text{out}}}

\IEEEoverridecommandlockouts

\begin{document}
\sloppy

\title{Adding a Helper Can  Totally Remove the Secrecy Constraints in Interference Channel}
\author{Jinyuan Chen and Fan Li
\thanks{Jinyuan Chen and Fan Li are with Louisiana Tech University, Department of Electrical Engineering, Ruston, LA 71272, USA (emails: jinyuan@latech.edu, fli005@latech.edu).  
} 
}

\maketitle
\pagestyle{headings}

\begin{abstract}

In many  communication channels,  secrecy constraints \emph{usually} incur a penalty in capacity, as well as generalized degrees-of-freedom (GDoF).  
In this work we show  an interesting observation that, adding a helper can \emph{totally} remove the   penalty in sum GDoF, for a two-user symmetric Gaussian interference channel.
For the interference channel where each transmitter sends a message to an intended receiver without secrecy constraints,  the sum GDoF  is a  well-known ``W'' curve, characterized by Etkin-Tse-Wang in 2008. 
If the secrecy constraints are imposed on this interference channel, where the message of each transmitter must be secure from the unintended receiver (eavesdropper),  then a GDoF penalty is incurred  and the secure sum GDoF is  reduced to a modified ``W'' curve, derived by Chen recently.
In this work we show that, by adding a helper into this interference channel with secrecy constraints, the \emph{secure} sum GDoF turns out to be a ``W'' curve, which is  the same as the sum GDoF of the setting without secrecy constraints.
The proposed scheme is based on the cooperative jamming and a careful signal design such that  the jamming signal of  the helper  is aligned at a specific direction and power level with the information signals of the transmitters, which allows to totally remove the penalty in GDoF due to the secrecy constraints. 
Furthermore, the estimation approaches of noise removal and signal separation due to the rational independence are used in the secure rate analysis.

\end{abstract}

\section{Introduction}

Since Shannon's  seminal work  in  1949  \cite{Shannon:49}, information-theoretic secrecy has been studied for many years in many communication channels, such as wiretap channels \cite{Wyner:75, CsiszarKorner:78, CH:78},  broadcast  channels \cite{LMSY:08, KTW:08, XCC:09, LLPS:10},    multiple access channels   
 and interference channels 
 \cite{TY:08cj,  LP:08,  TekinYener:08d, HY:09, KG:15, HKY:13, LMSY:08,LBPSV:09,PDT:09, LYT:08,YTL:08, KGLP:11, LLP:11,XU:14, MDHS:14, MM:14o, XU:15,  MM:16, GTJ:15, GJ:15, allerton:16,  MXU:17, ChenAllerton:18, ChenIC:18}. 
In the pioneer work by Wyner in 1975 \cite{Wyner:75},  the notion of secure capacity was introduced via  a  wiretap channel with secrecy constraint. For this wiretap channel,  Wyner showed that there is a penalty in capacity due to the secrecy constraint.
Later on, it has been shown that this insight also holds true for the other channels 
(see, e.g., \cite{LMSY:08, LBPSV:09, HKY:13, KGLP:11, XU:14, XU:15, MM:16,GTJ:15, ChenAllerton:18}).

This work focuses on the secure sum generalized degrees-of-freedom (GDoF, a capacity approximation) of a  two-user symmetric Gaussian interference channel, and shows an interesting observation that the penalty in GDoF due to the secrecy constraints can be \emph{totally} removed by adding a helper.  For the symmetric Gaussian interference channel \emph{without} secrecy constraints, the sum GDoF is characterized as a well-known ``W'' curve (cf.~\cite{ETW:08}).  If the secrecy constraints are imposed on this interference channel, where the message of each transmitter must be secure from the unintended receiver (eavesdropper),  the work in \cite{ChenIC:18} showed that the secure sum GDoF is then reduced to a modified ``W'' curve (see Fig.~\ref{fig:IChGDoF}). It reveals that there is a penalty in GDoF in a large regime due to the secrecy constraints.  Interestingly, this work shows that  adding a helper can totally remove the penalty in GDoF due to the secrecy constraints (see Fig.~\ref{fig:IChGDoF}).  In other words,  by adding a helper, the secure sum GDoF of the interference channel derived in this work  is exactly the same as the sum GDoF  of the interference channel \emph{without} secrecy constraints.

The proposed scheme is based on cooperative jamming  and signal  alignment.  Specifically, the helper sends a cooperative jamming signal  at a specific direction and power level to confuse the eavesdroppers, while keeping legitimate receivers' abilities to decode their desired messages.  
The role of helper(s) in secure communications has been studied extensively in the literature (see \cite{TY:08cj,LMSY:08, XU:14, XU:15, FW:16} and references therein).  The helper can be the node transmitting its own confidential message (cf.~\cite{TY:08cj, LMSY:08, XU:14, XU:15}), as well as the node without transmitting any message  (cf.~\cite{ XU:14,FW:16}). 
The helper considered in our setting refers to the latter case, which is able to totally remove the penalty in GDoF due to the secrecy constraints.
For the proposed scheme,  the estimation approaches of noise removal and signal separation due to  rational independence  (cf.~\cite{MGMK:14}) are used in the secure rate analysis.

\begin{figure}[t!]
\centering
\includegraphics[width=8cm]{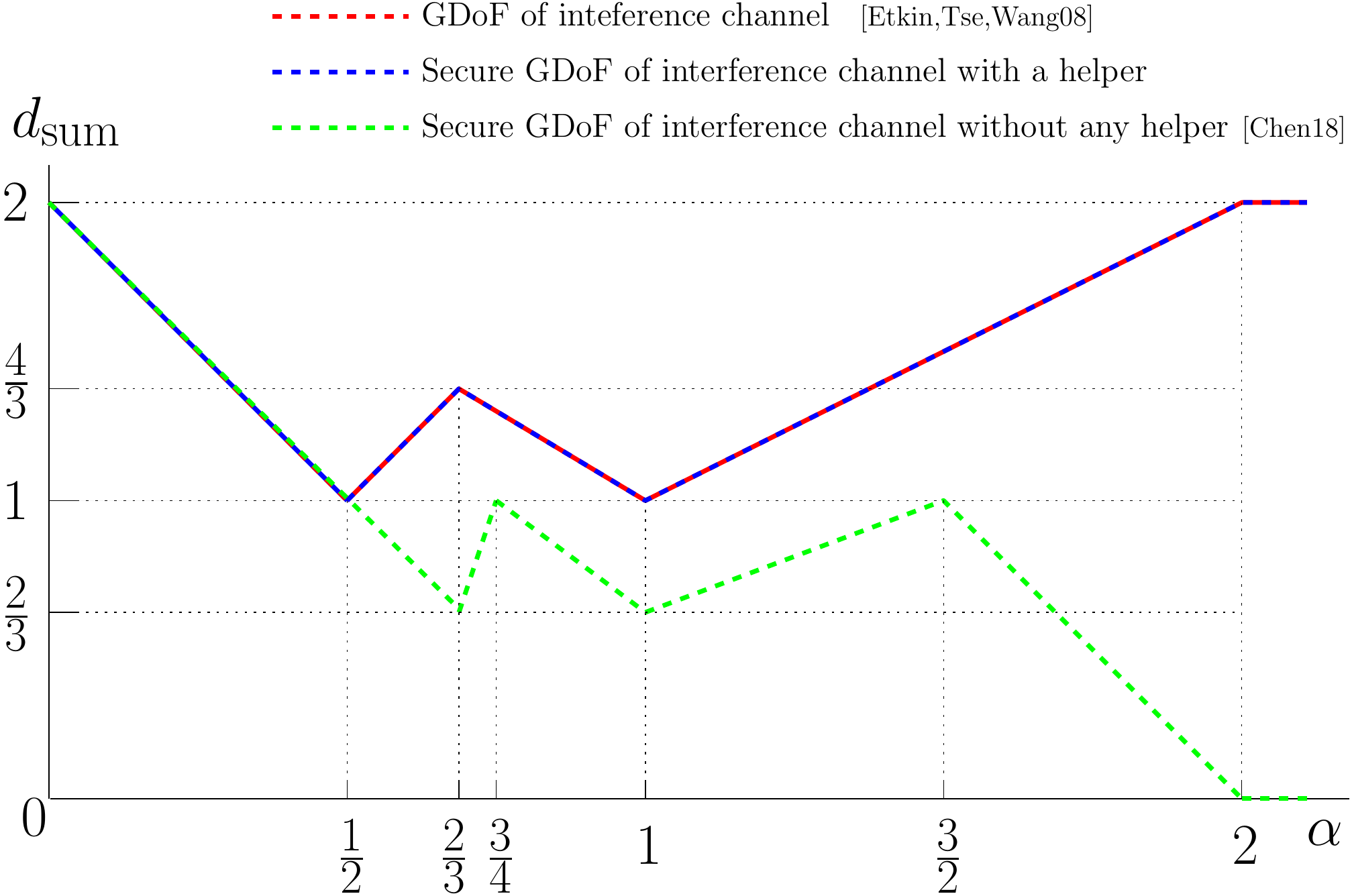}
\caption{Secure sum GDoF vs. $\alpha$ for the two-user \emph{symmetric} Gaussian interference channel with a helper, where $\alpha$ corresponds to the interference-to-signal ratio (all link strengths in decibel scale).   Note that  the secure sum GDoF of the interference channel with a helper, derived in this work (cf.~Theorem~\ref{thm:IChGDoF}),  is the same as the sum GDoF  of the interference channel \emph{without} secrecy constraints (cf.~\cite{ETW:08}). Without any helper, there is a penalty in GDoF due to the secrecy constraints (cf.~\cite{ChenIC:18}).}
\label{fig:IChGDoF}
\end{figure}

The remainder of this work is organized as follows. 
Section~\ref{sec:system} introduces the  system model.  
Section~\ref{sec:mainresult} presents the  main result of this work. 
Section~\ref{sec:example}  describes the proposed  scheme via an  example. 
The achievability proof  is given in  Sections~\ref{sec:CJGau}-\ref{sec:rateerror341}  and the appendices, while the converse proof is given in Section~\ref{sec:converse}.
Section~\ref{sec:conclusion} concludes this work.  
Throughout this work, $\Imu(\bullet)$, $\Hen(\bullet)$ and $\hen(\bullet)$ denote the mutual information, entropy and differential entropy,  respectively.  
     $\Zc$,  $\Zc^+$  and $\Rc$  denote the sets of integers, positive integers  and real numbers, respectively.   
        $o(\bullet)$ comes from the standard Landau notation, where  $f(x)=o(g(x))$ implies that $\lim_{x \to \infty} f(x)/g(x) =0$.  
Logarithms are in base~$2$.

\section{System model   \label{sec:system} }

In this setting, we consider a two-user Gaussian interference channel with a helper and confidential messages (see Fig.~\ref{fig:IChelper}). The output of the channel  at receiver~$k$ at time~$t$ is 
\begin{align}
y_{k}(t) &= \sum_{\ell=1}^{3}  2^{m_{k\ell}} h_{k\ell} x_{\ell}(t) +z_{k}(t), \quad \quad k=1,2  \label{eq:channelG101} 
\end{align}
$t=1,2, \cdots, \bln $, where  $x_{\ell}(t)$ is the normalized channel input  at transmitter~$\ell$ under the  power constraint  $\E |x_{\ell}(t)|^2 \leq 1$; $z_k(t) \sim \mathcal{N}(0, 1)$ is the additive white Gaussian noise at receiver~$k$;  $m_{k\ell}$ is a nonnegative integer; 
and  $h_{k\ell} \in (1, 2]$ is a normalized channel coefficient,  for $k =1,2$ and  $\ell = 1,2, 3$.  In our setting, transmitter~3 is the helper. 
By following the convention in \cite{ChenIC:18}, we  let   $P \defeq \max_{k}\{ 2^{2 m_{kk}} \}$, and define  
\begin{align}
\alpha_{k\ell} &\defeq  \frac{\log 2^{m_{k\ell}}}{ \frac{1}{2}\log P}  \quad \quad   k =1,2, \quad  \ell = 1,2, 3.  \label{eq:alpha11} 
\end{align}
Then, we can rewrite the channel model in \eqref{eq:channelG101} as
\begin{align}
y_{k}(t) &= \sum_{\ell =1}^{3}  \sqrt{P^{\alpha_{k\ell}}} h_{k\ell} x_{\ell}(t) +z_{k}(t), \quad \quad k=1,2  \label{eq:channelG} 
\end{align}
where  $\alpha_{k\ell} \geq 0$ denotes the \emph{channel strength} of the link between transmitter~$\ell$ and receiver~$k$. 
In the rest of this work, we will consider the channel model in \eqref{eq:channelG}. 
It is assumed that all the channel parameters $\{\alpha_{k\ell}, h_{k\ell}  \}_{k, \ell}$  are available at each node. 
For the \emph{symmetric} case, it is assumed that 
\[ \alpha_{11}= \alpha_{22}=1, \quad  \alpha_{21}= \alpha_{12}= \alpha_{13}= \alpha_{23} = \alpha,  \quad   \alpha \geq 0.\]

\begin{figure}[t!]
\centering
\includegraphics[width=7cm]{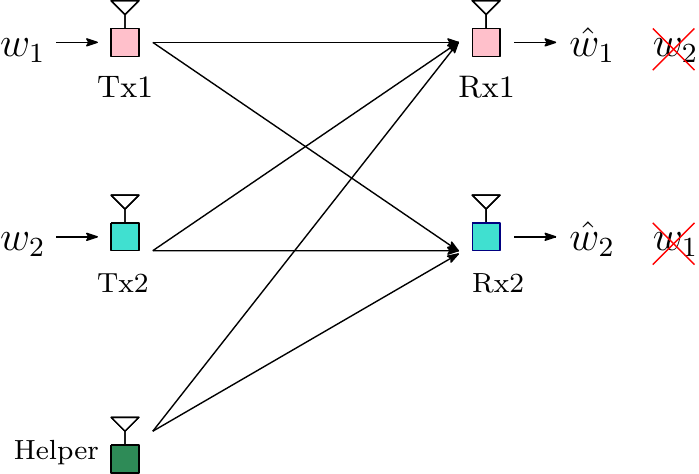}
 \vspace{-.05 in}
\caption{Two-user interference channel with a helper and confidential messages.}
\label{fig:IChelper}
\end{figure}

In this interference channel, a confidential message $w_k$  is sent from transmitter~$k$  to  receiver~$k$,  where  $w_k$  is  uniformly chosen from a set $\mathcal{W}_k \defeq \{1,2,\cdots, 2^{\bln R_k}\}$, for $k=1,2$. 
At  transmitter~$k$, a  stochastic function 
\[f_k: \mathcal{W}_k  \times   \mathcal{W}'_{k}   \to   \mathcal{R} ^{\bln}, \quad  \quad k=1,2\] is used to map the message $w_k \in \mathcal{W}_k$  to a transmitted codeword  $ x_k^{\bln}  = f_k(w_k, w'_k )   \in  \mathcal{R} ^{\bln}$, where $w'_k \in  \mathcal{W}'_k$  denotes the randomness in this mapping, and $w'_k$ is available at transmitter~$k$ only. 
At transmitter~$3$ (the helper), the following function \[f_3:   \mathcal{W}'_{3}   \to   \mathcal{R} ^{\bln} \] is used to generate $x_3^{\bln}  = f_3(w'_k ) $, where the random variable $w'_3 \in \mathcal{W}'_3$ is available at transmitter~$3$ only.
We assume that the random variables $\{w_1, w'_1 , w_2, w'_2 ,   w'_3\}$ are mutually independent. 
A secure rate pair $(R_1, R_2)$ is said to be achievable  if for any $\epsilon >0$ there exists a sequence of $\bln$-length codes such that each receiver can decode its own message reliably 
and the messages are kept secret such that 
 \begin{align}
 \Imu(w_1; y_{2}^{\bln})  \leq  \bln \epsilon , \quad
\Imu(w_2; y_{1}^{\bln})   \leq  \bln \epsilon.  
 \end{align} 
The secure capacity region, denoted by $C$,  is the closure of the set of all achievable  $(R_1,  R_2)$ secure rate pairs.
We define the secure sum capacity as: 
 \begin{align}
 C_{\text{sum}} \defeq \sup \big\{ R_1 + R_2 |  \  (R_1, R_2) \in C  \big\} .  \label{eq:defGDoFsum}
  \end{align}
We also define the secure sum GDoF  as   
 \begin{align}
 d_{\text{sum}}  \defeq   \lim_{P \to \infty}   \frac{C_{\text{sum}}}{ \frac{1}{2} \log P}.  \label{eq:defGDoF}
 \end{align}
Note that the GDoF  is  a  form of capacity approximation. It is more general than degrees-of-freedom (DoF),  as the latter considers only a specific case with $\alpha_{k\ell} =1,  \forall k, \ell $.

\section{Main result  \label{sec:mainresult}}

This section provides   the  main result of this work,  for the secure communication over a  two-user symmetric  Gaussian interference channel with a helper.

\begin{theorem}  \label{thm:IChGDoF}
Considering the two-user symmetric  Gaussian interference channel with a helper defined in Section~\ref{sec:system}, for almost all the channel coefficients  $\{h_{k\ell}\} \in (1, 2]^{2\times 3}$,  the optimal secure sum GDoF  is characterized as 
\begin{subnumcases}
{ d_{\text{sum}}  =} 
     2(1- \alpha)    &    for   \ $ 0 \leq \alpha \leq  \frac{1}{2}$            			\label{thm:GDoFICh1} \\
     2\alpha  &  for \ $\frac{1}{2}  \leq \alpha \leq  \frac{2}{3}$           		\label{thm:GDoFICh2} \\ 
        2(1 -  \alpha / 2)  &  for  \   $\frac{2}{3}  \leq \alpha \leq  1$    			\label{thm:GDoFICh3} \\ 
            \alpha  &  for  \  $1  \leq  \alpha \leq   2$    							\label{thm:GDoFICh4} \\ 
                                               2  &  for  \   $  \alpha  \geq 2$ .  				\label{thm:capacitydet5}
\end{subnumcases}
\end{theorem}

In this setting defined in Section~\ref{sec:system},  Theorem~\ref{thm:IChGDoF} reveals that  adding a helper can totally remove the secrecy constraints, in the sense that  the secure sum GDoF of the interference channel with a helper and with secrecy constraints is the same as the sum GDoF  of the interference channel without any helper and without secrecy constraints (a ``W'' curve, see \cite{ETW:08} and Fig.~\ref{fig:IChGDoF}).  The optimal secure sum GDoF is achieved by a cooperative jamming scheme. 
The achievability is described in Section~\ref{sec:CJGau}. The converse proof is provided in Section~\ref{sec:converse}.
Before providing the achievability, we describe the proposed scheme via an example in the following section.

\section{Scheme example  \label{sec:example} }

In this section, we will describe the outline of the proposed scheme via an example. 
We will focus on the specific case with $\alpha = 3/4 $,  for  the two-user \emph{symmetric} Gaussian  interference channel with secrecy constraints and with a helper defined in Section~\ref{sec:system}.
Note that, for  the two-user \emph{symmetric} Gaussian  interference channel without any secrecy constraints, the sum GDoF is $5/4$ when $\alpha = 3/4$ (cf.~\cite{ETW:08}).
If the secrecy constraints are imposed on this interference channel, where the message of each transmitter must be secure from the unintended receiver,   the \emph{secure} sum GDoF is  reduced to $1$ when $\alpha = 3/4$  (cf.~\cite{ChenIC:18}). 
Therefore, there is a penalty in GDoF due to the secrecy constraints.  Interestingly, we will shows that  adding a helper can totally remove the penalty in GDoF due to the secrecy constraints, that is,   by adding a helper the secure sum GDoF is increased to $5/4$, which  is exactly the same as the sum GDoF  of the interference channel \emph{without} secrecy constraints.
For the proposed scheme,  pulse  amplitude modulation (PAM)   and signal alignment will be used in the signal design, and the estimation approaches of noise removal and signal separation will be used in the rate analysis.  The scheme is motivated by the  cooperative jamming  scheme proposed in \cite{ChenIC:18} for a different setting, i.e., a two-user interference channel with confidential messages but without any helper.

\begin{figure}[t!]
\centering
\includegraphics[width=14cm]{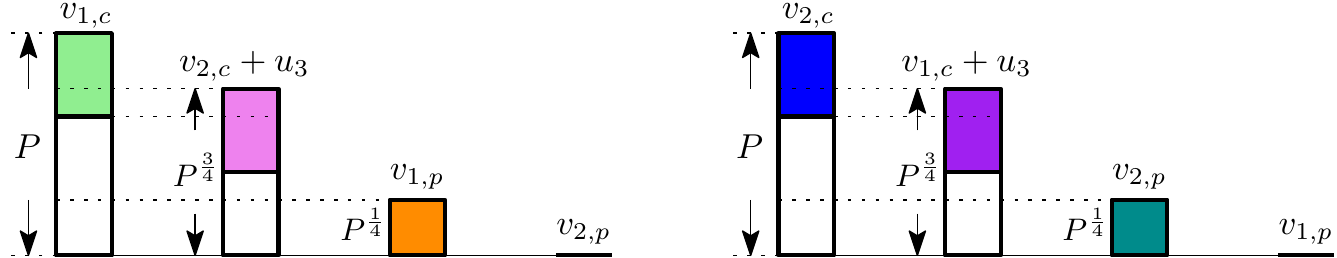}
 \vspace{-.05 in}
 \caption{Rate (GDoF)  and power description at receivers~1 and 2  when $\alpha =  3/4$.}
\label{fig:allocation_c34}
\end{figure}

For this case with $\alpha = 3/4$, the transmitted signals are designed as (removing the time index):
 \begin{align}
   x_1  = &     h_{2 3}  h_{1 2} v_{1,c} +   \sqrt{P^{ -  3/4}}    \cdot  h_{2 3}  h_{1 2} v_{1,p}      \label{eq:xvkkk23100e}  \\
   x_2  = &     h_{1 3}  h_{2 1} v_{2,c} +   \sqrt{P^{ -  3/4}}    \cdot  h_{1 3}  h_{2 1} v_{2,p}   \label{eq:xvkkk23111e}  \\
      x_3 = &     h_{1 2}  h_{21}  u_{3},   \label{eq:xvkkk1341e}  
 \end{align}
 where  $v_{k,c}$ and $v_{k,p}$ are two signals that carry the message of transmitter~$k$, for $k=1,2$, and     $u_{3}$ is the jamming signal from transmitter~3.    
 The random variables $\{v_{1,c}, v_{1,p}, v_{2,c}, v_{2,p}, u_{3}\}$ are \emph{mutually independent}.  Specifically, for $\Omega (\xi,  Q)  \defeq   \{ \xi \cdot a :   \    a \in  \Zc  \cap [-Q,   Q]   \}$  denoting the PAM constellation set,  the above random variables   are  \emph{independently}  and \emph{uniformly}  drawn from their PAM constellation sets, given as 
 \begin{align}
   v_{k,c}      &  \in    \Omega ( \xi =  2\gamma \cdot \frac{ 1}{Q} ,  \quad  Q =  P^{ \frac{  3/8  - \epsilon }{2}} )  \label{eq:constellationGsym1e}   \\ 
   v_{k,p}      &  \in    \Omega ( \xi = \gamma \cdot \frac{ 1}{Q} ,   \    \quad  Q = P^{ \frac{ 1/4  - \epsilon  }{2}} )    \label{eq:constellationGsym2e}    
   \\
      u_{3}      &  \in    \Omega ( \xi =  2\gamma \cdot \frac{ 1}{Q} ,      \quad  Q =  P^{ \frac{ 3/8  - \epsilon  }{2}} )  \label{eq:constellationGsym1ue} 
 \end{align}
 respectively, for $k=1,2$, where $\epsilon>0$ can be made arbitrarily small,   and  $\gamma$ is a finite constant such that  $\gamma  \in \bigl(0, \frac{1}{8\sqrt{2}}\bigr]$. 
 Based on our signal design,   $v_{k,c}$ carries  $ 3/8 $ GDoF and  $v_{k,p}$ carries  $1/4$  GDoF, that is, $ \Hen(v_{k,c}) =  \frac{3/8}{2}\log  P + o(\log P)$ and  $ \Hen(v_{k,p}) =  \frac{1/4}{2}\log  P + o(\log P)$,  when $\epsilon \to 0$,  for $k=1,2$. 
The above signal design satisfies the average power constraints, $\E |x_1|^2 \leq 1$, $\E |x_2|^2 \leq 1$ and $\E |x_3|^2 \leq 1$, which will be proved in the next section in detail. 
 Then, the received signals at the receivers~1 and~2 are given as (without time index)
\begin{align}
y_{1} &=    \sqrt{P} h_{11}    h_{23} h_{12} v_{1,c}  +    \sqrt{P^{ 1/4}} h_{11}    h_{23} h_{12}  v_{1,p}   
\non\\& \quad  +    \sqrt{P^{ 3/4 }} h_{12}h_{21} h_{13}  (   v_{2,c}  +  u_3) +   h_{12}h_{21} h_{13} v_{2,p}   +  z_{1}    \label{eq:yvk1341e}  \\
y_{2} &=       \sqrt{P} h_{22}h_{13}h_{21}  v_{2,c}    +   \sqrt{P^{ 1/4}}h_{22}h_{13}h_{21}  v_{2,p}  
 \non\\ &\quad  +    \sqrt{P^{ 3/4 }}   h_{21}  h_{12} h_{23}(    v_{1,c}  +  u_3)    +   h_{21}  h_{12}  h_{23} v_{1,p}   +  z_{2} .    \label{eq:yvk2341e}   
\end{align}
At receiver~1,  the information signal $v_{2,c}$ of transmitter~2 is \emph{aligned} with the jamming signal $u_3$ of the helper. 
At receiver~2,  the information signal $v_{1,c}$ of transmitter~1 is \emph{aligned} with the jamming signal $u_3$ of the helper. 
As we will see later on, the penalty in GDoF due to the secrecy constraint will be minimized with this signal alignment design. 
Fig.~\ref{fig:allocation_c34} describes the  rate (GDoF) and power of some signals at receivers~1 and 2  when $\alpha = 3/4$.

As we will discuss in detail in the next section,  the  \emph{secure} rate pair 
 \begin{align}
 R_1  &=    \Imu(v_1; y_1) - \Imu(v_1; y_2 | v_2 )    \label{eq:exampleR1e}     \\
R_2  &=    \Imu(v_2; y_2) - \Imu(v_2; y_1 | v_1 )     \label{eq:exampleR2e}     
 \end{align}
is achievable by the proposed scheme with a careful codebook design and message mapping.
Let us begin with the secure rate $R_1$ expressed in \eqref{eq:exampleR1e}. On one hand, we expect the term $ \Imu(v_1; y_1)$ to be sufficiently large; on the other hand, we expect the term $ \Imu(v_1; y_2 | v_2 ) $ to be sufficiently small.
Let $ \hat{v}_{1,c}$ and  $\hat{v}_{1,p}  $ be the  estimates for  $v_{1,c}$ and $v_{1,p}$ respectively from $y_1$,  and let $ \text{Pr} [  \{ v_{1,c} \neq \hat{v}_{1,c} \} \cup  \{ v_{1,p} \neq \hat{v}_{1,p} \}  ] $ denote the corresponding estimation error probability.
Then the term $\Imu(v_1; y_1)$ can be lower bounded by
 \begin{align}
  &\Imu(v_1; y_1)     \non\\
  \geq &  \Imu(v_1; \hat{v}_{1,c}, \hat{v}_{1,p})  \label{eq:rateAWGNIC1202e}     \\
  =  & \Hen(v_1) -   \Hen(v_1  |  \hat{v}_{1,c}, \hat{v}_{1,p})    \non    \\
    \geq &   \Hen(v_1) -       \bigl( 1+    \text{Pr} [  \{ v_{1,c} \neq \hat{v}_{1,c} \} \cup  \{ v_{1,p} \neq \hat{v}_{1,p} \}  ] \cdot \Hen(v_{1}) \bigr)  \label{eq:rateAWGNIC1404e}     \\
         =    &  \bigl( 1 -   \text{Pr} [  \{ v_{1,c} \neq \hat{v}_{1,c} \} \cup  \{ v_{1,p} \neq \hat{v}_{1,p} \}  ] \bigr)   \cdot \Hen(v_{1})  - 1  \label{eq:rateAWGNIC1405e}     
 \end{align}
where \eqref{eq:rateAWGNIC1202e} results from the Markov chain $v_1 \to y_1 \to  \{\hat{v}_{1,c}, \hat{v}_{1,p} \}$;
\eqref{eq:rateAWGNIC1404e} results from Fano's inequality.
Note that, based on our signal design,  $\{ v_{1,p}, v_{1,c}\}$ can be reconstructed from $v_{1}$,  and vice versa, because $v_{1}$ takes the form of $v_{1} =   v_{1,c} +   \sqrt{P^{ - 3/4}}    \cdot  v_{1,p} $  and  the minimum of $v_{1,c}/2$ is bigger than the maximum of $\sqrt{P^{ - 3/4}}    \cdot  v_{1,p}$.  
In this case  $v_{1,c}$ carries  $ 3/8 $ GDoF and  $v_{1,p}$ carries  $1/4$  GDoF, which implies that $v_{1}$ carries $5/8$ GDoF. The term $\Hen(v_{1})$ in \eqref{eq:rateAWGNIC1405e} then becomes  
\begin{align}
\Hen(v_{1}) =  \frac{5/8}{2}\log  P + o(\log P)  .  \label{eq:error82676e}
\end{align}
In Sections~\ref{sec:CJGau}  and \ref{sec:rateerror341} we will prove that  $v_{1,c}$ and $v_{1,p}$ designed in \eqref{eq:constellationGsym1e}-\eqref{eq:constellationGsym1ue} can be estimated from $y_1$ by using the estimation approaches of noise removal and signal separation,  and the corresponding estimation error probability is 
\begin{align}
\text{Pr} [  \{ v_{1,c} \neq \hat{v}_{1,c} \} \cup  \{ v_{1,p} \neq \hat{v}_{1,p} \}  ]  \to 0         \quad \text {as}\quad  P\to \infty    \label{eq:error9266e}
\end{align}
 for \emph{almost} all the channel coefficients  $\{h_{k\ell}\} \in (1, 2]^{2\times 3}$.  
The proof outline of  \eqref{eq:error9266e} will be provided later on in this section. 
By incorporating \eqref{eq:error82676e} and \eqref{eq:error9266e} into \eqref{eq:rateAWGNIC1405e}, it gives 
 \begin{align}
  \Imu(v_1; y_1)   
    \geq    \frac{5/8}{2}\log  P + o(\log P)   \label{eq:rateAWGNIC28667e}     
 \end{align}
for almost all the channel coefficients  $\{h_{k\ell}\} \in (1, 2]^{2\times 3}$.

On the other hand,   $\Imu(v_1; y_2 | v_2 )$ can be bounded  as
   \begin{align}
  \Imu(v_1; y_2 | v_2 )     \leq  & \log (2\sqrt{65})  .  \label{eq:rateAWGNIC9562341e}  
 \end{align}
 The detailed proof is provided in the next section. 
 This term $\Imu(v_1; y_2 | v_2 )$ can be considered as a penalty term in the secure rate $R_1$.  The result in \eqref{eq:rateAWGNIC9562341e} reveals that this penalty is sufficiently small, based on our  careful signal design. In the proposed scheme,  the jamming signal of  the helper  is aligned at a specific direction and power level with the information signals of the transmitters,  which minimizes  the  rate penalty  due to the secrecy constraints. 
Finally,  with the results in \eqref{eq:rateAWGNIC28667e}   and \eqref{eq:rateAWGNIC9562341e}, the secure rate $R_1$ is lower bounded by 
\begin{align}
R_1   =    \Imu(v_1; y_1)\! -\! \Imu(v_1; y_2 | v_2 )   \!\geq\!    \frac{5/8}{2}\log  P \!+\! o(\log P)   \non 
\end{align}
and similarly, the secure rate $R_2$ is lower bounded by  $R_2 \geq     \frac{5/8}{2}\log  P + o(\log P)$, for  almost all the channel coefficients  $\{h_{k\ell}\} \in (1, 2]^{2\times 3}$.
Then the secure GDoF pair $(d_1 = 5/8,  d_2 =  5/8 )$ can be achieved by the proposed cooperative jamming scheme for  almost all the channel coefficients  in  this case with  $ \alpha =3/4$.

In the following we  provide the proof outline of  \eqref{eq:error9266e}.   
After some manipulations,   $y_1$ in \eqref{eq:yvk1341e}  can be rewritten as 
\begin{align}
y_{1}  & = P^{{ \frac{ 3/8 + \epsilon }{2}}} \cdot  2\gamma \cdot  \underbrace{(\sqrt{P^{ 1/4}}  g_0 q_0 + g_1 q_1 )}_{\defeq x_{s}}  + \tilde{z}_{1}  \non
\end{align}
where $x_{s} \defeq \sqrt{P^{ 1/4 }}  g_0 q_0 + g_1 q_1 $,   $\tilde{z}_{1}   \defeq     \sqrt{P^{1/4}} h_{11}    h_{23} h_{12}  v_{1,p}  +    h_{12}h_{21} h_{13} v_{2,p}   +  z_{1}$, 
 $ g_0\defeq  h_{11}    h_{23} h_{12}$,   $g_1\defeq   h_{12}h_{21} h_{13} $
   $q_0  \defeq  \frac{Q_{\max}}{2\gamma} \cdot   v_{1,c} $,   $q_1  \defeq  \frac{Q_{\max}}{2\gamma} \cdot   (   v_{2,c}  +  u_{3})$,  and $ Q_{\max} \defeq P^{ \frac{ 3/8 - \epsilon }{2}} $.
  One important step is to estimate $x_{s}$ from $y_{1}$  by treating other signals as noise. This step is called as noise removal.  After correctly decoding  $x_{s} $, we can recover $q_0$  and $q_1$ from $x_{s} $  because $g_0$ and $g_1$ are rationally independent. This step is called as signal separation (cf.~\cite{MGMK:14}).
As we will show in Section~\ref{sec:rateerror341}, the average power of the virtual noise $\tilde{z}_{1} $ is bounded by 
\[ \E |\tilde{z}_{1}|^2 \leq  \kappa P^{1/4} \]
  for a finite constant $\kappa$. 
  To estimate $x_{s}$ from $y_{1}$,  we  show in Section~\ref{sec:rateerror341} that  the minimum distance of $x_{s}$, denoted by $d_{\min}$,  is  bounded by
 \begin{align}
d_{\min}    \geq   \delta P^{- 1/16}    \label{eq:distancegeqe}
 \end{align}
for all  the channel  coefficients $\{h_{k\ell}\} \in (1, 2]^{2\times 3} \setminus \Ho$, where $\Ho \subseteq (1,2]^{2\times 3}$ is an outage set and  the Lebesgue measure of this outage set, denoted by $\mathcal{L}(\Ho)$,  satisfies  
 \begin{align}
\mathcal{L}(\Ho) \leq 1792 \delta   \cdot     P^{ - \frac{ \epsilon  }{2}}   \label{eq:lebesque11e}
 \end{align}
 for some constants $\delta \in (0, 1]$ and  $\epsilon >0$.
 The proofs of  \eqref{eq:distancegeqe} and \eqref{eq:lebesque11e} build on the conclusion of Lemma~\ref{lm:NMb2} (see below).
 The results of \eqref{eq:distancegeqe} and \eqref{eq:lebesque11e} reveal  that,  the minimum distance of $x_{s}$  is sufficiently large
for almost all  the channel  coefficients in the regime of large $P$, that is $\mathcal{L}(\Ho)  \to 0$ as $P\to \infty$.
At this point  we can decode the sum $x_{s}  =   \sqrt{P^{ 1/4 }}  g_0 q_0 + g_1 q_1 $ 
from $y_1$ by treating other signals as noise (noise removal), with vanishing error probability. 
After that we can recover $q_0  =  \frac{Q_{\max}}{2\gamma} \cdot   v_{1,c} $, as well as $v_{1,c} $,  from $x_{s} $  because $g_0$ and $g_1$ are rationally independent (rational independence).  By removing the decoded $x_{s} $ from  $y_1$, $v_{1,p} $ can be estimated with vanishing error probability, resulting in  
 \begin{align}
 \text{Pr} [  \{ v_{1,c} \neq \hat{v}_{1,c} \} \cup  \{ v_{1,p} \neq \hat{v}_{1,p} \}  ]  \to 0         \quad \text {as}\quad  P\to \infty   \non
 \end{align}
 for   almost all  the channel  coefficients in the regime of large $P$. The lemma used in the proofs of   \eqref{eq:distancegeqe} and \eqref{eq:lebesque11e} is given as follows.

\begin{lemma}   \label{lm:NMb2}
Let $\beta \in (0,1]$,  $\tau  \in \Zc^+$ and $ \tau >1$, and $A_0, A_1, Q_0, Q_1  \in \Zc^+$.  Define the event
\begin{align}
  B(q_0, q_1)  \defeq \{ (g_0, g_1)  \in (1, \tau]^2 :  |  A_0 g_0 q_0 + A_1 g_1 q_1  | < \beta \}       \label{eq:lemmabound0098}  
\end{align}
and set 
\begin{align}
B  \defeq   \bigcup_{\substack{ q_0, q_1 \in \Zc:  \\  (q_0, q_1) \neq  0,  \\  |q_k| \leq Q_k  \ \forall k }}  B(q_0, q_1) .       \label{eq:lemmabound125}
\end{align}
Then the Lebesgue measure of $B $, denoted by $\Lc (B)$,  is bounded by
\begin{align*}
\Lc (B )  \leq   8 (\tau -1) \beta \min\{ \frac{ Q_1 Q_0}{A_1},  \frac{Q_0 Q_1}{A_0},   \frac{Q_0 \tau}{A_1},   \frac{Q_1 \tau}{A_0}\}.
\end{align*}
\end{lemma}

\begin{proof}
See Section~\ref{sec:lm:NMb2}.
\end{proof}

\section{Achievability \label{sec:CJGau} }

In this  section, for the two-user \emph{symmetric} Gaussian  interference channel defined in Section~\ref{sec:system},   we provide  a cooperative jamming scheme,  focusing on the regime of  $\alpha > 1/2$.  
Note that when $0 \leq  \alpha \leq 1/2$, the optimal secure sum GDoF $d_{\text{sum}} = 2(1- \alpha)$ is achievable by a scheme without any helper and without cooperative jamming (cf.~\cite{ChenIC:18}). 
In the proposed scheme,   pulse  amplitude modulation and signal alignment  will be used.  
In what follows, we describe the details of codebook generation and  signal mapping, PAM constellation and  signal alignment, and  secure rate analysis.

\subsubsection{Codebook generation and signal mapping}
At transmitter~$k$, $k=1,2$, it generates a codebook given as
  \begin{align}
     \mathcal{B}_{k} \defeq \Bigl\{  v^{\bln}_k (w_k,  w_k'):  \  w_k \in \{1,2,\cdots, 2^{\bln R_k}\},   
      w_k' \in \{1,2,\cdots, 2^{\bln R_k'}\}   \Bigr\}     \label{eq:code2341J}
     \end{align}
where $v^{\bln}_k $ are the codewords. All the elements of the codewords are independent and identically generated according to a distribution that will be defined later on. 
In \eqref{eq:code2341J}, $w_k'$ is the confusion message that is used to guarantee the security of  the confidential message $w_k$.
  $R_k$ and  $R_k'$ are the rates  of  $w_k$  and $w_k'$, respectively, which will be defined later on (see \eqref{eq:Rk623J} and  \eqref{eq:Rk623bJ}).
To transmit the message $w_k$, transmitter~$k$ selects a sub-codebook  $\mathcal{B}_{k}( w_k) $ that is   defined  as 
\[   \mathcal{B}_{k} (w_k)  \defeq \bigl\{ v^{\bln}_k (w_k,  w_k'): \  w_k' \in \{1,2,\cdots, 2^{\bln R_k'}\}   \bigr\},  \quad k=1,2  \]
and  then  \emph{randomly} selects a codeword $v^{\bln}_k$ from $\mathcal{B}_{k}( w_k) $ according to a uniform distribution.
In this scheme, the chosen codeword $v^{\bln}_k$ will be mapped  to the channel input such that 
 \begin{align}
  x_k (t) =   h_{\ell 3}  h_{k \ell}  v_k (t)      \label{eq:xvkkk}      
   \end{align}
  for $ \ell, k =1,2, \ell \neq k$, and  $t=1,2, \cdots, \bln$, where $v_k (t)$ is the $t$th element of the codeword $v^{\bln}_k$.

\subsubsection{PAM constellation and signal alignment}

Specifically,  each element of the codeword at transmitter $k$, $k=1,2$,  is designed to take the following form
 \begin{align}
   v_{k}  =     \sqrt{P^{ - \beta_{v_{k,c}}}}  \cdot  v_{k,c} +   \sqrt{P^{ - \beta_{v_{k,p}}}}    \cdot  v_{k,p}   \label{eq:xvk}  
 \end{align}
 which implies that   the  channel input in~\eqref{eq:xvkkk} can be rewritten as 
 \begin{align}
  x_k  = &    \sqrt{P^{ - \beta_{v_{k,c}}}}  \cdot h_{\ell 3}  h_{k \ell} v_{k,c} +   \sqrt{P^{ - \beta_{v_{k,p}}}}    \cdot  h_{\ell 3}  h_{k \ell} v_{k,p} \label{eq:xvkkk1}  
 \end{align}
(removing the time index for simplicity) for  $ \ell, k =1,2, \ell \neq k$.
At the helper (transmitter~3), it sends a cooperative jamming signal  designed as  
    \begin{align}
  x_3  = &    \sqrt{P^{ - \beta_{u_{3}}}}  \cdot h_{1 2}  h_{21}  u_{3}.      \label{eq:u3def}  
 \end{align}
In \eqref{eq:xvkkk1} and \eqref{eq:u3def}, $v_{k,c}$, $v_{k,p}$ and $u_{3}$ are   \emph{independent} random variables   which  are \emph{uniformly}  drawn from their PAM constellation sets
 \begin{align}
   v_{k,c}      &  \in    \Omega ( \xi =  \gamma_{v_{k,c}} \cdot \frac{ 1}{Q} ,   \   Q =  P^{ \frac{ \lambda_{v_{k,c}} }{2}} )  \label{eq:constellationGsym1}   \\ 
   v_{k,p}      &  \in    \Omega ( \xi =  \gamma_{v_{k,p}} \cdot \frac{ 1}{Q} ,   \   Q = P^{ \frac{  \lambda_{v_{k,p}} }{2}} )    \label{eq:constellationGsym2}    
   \\
      u_{3}      &  \in    \Omega ( \xi =  \gamma_{u_{3}} \cdot \frac{ 1}{Q} ,      \quad  Q =  P^{ \frac{ \lambda_{u_{3}} }{2}} )  \label{eq:constellationGsym1u} 
 \end{align}
 respectively, for $k=1,2$,  
 and  $\gamma_{v_{k,c}}, \gamma_{v_{k,p}}$ and $\gamma_{u_{3}}$ are some finite constants designed as
  \begin{align}
  \gamma_{v_{1,c}} \!=  \gamma_{v_{2,c}} \!=\gamma_{u_3} \!=  2\gamma_{v_{1,p}} \!=  2\gamma_{v_{2,p}} \!= 2\gamma  \in \bigl(0, \frac{1}{4\sqrt{2}}\bigr].  \label{eq:gammadef} 
 \end{align}
  Table~\ref{tab:para} provides the designed parameters $\{\beta_{v_{k,c}}, \beta_{v_{k,p}}, \beta_{u_{3}}, \lambda_{v_{k,c}},  \lambda_{v_{k,p}}, \lambda_{u_{3}}\}_{k=1,2}$  for different cases of $\alpha$ in the proposed scheme.  
Based on our signal design (see \eqref{eq:xvkkk1}-\eqref{eq:gammadef}),  it can be checked that  the power constraint  $\E |x_k|^2 \leq 1$ is satisfied for $k=1,2, 3$.  For example,  since  $v_{1,c}$  and $v_{1,p}$ are uniformly drawn from $ \Omega (\xi =  2\gamma \cdot \frac{ 1}{Q},  Q=P^{ \frac{ \lambda_{v_{1,c}} }{2}})$ and $ \Omega (\xi =  \gamma \cdot \frac{ 1}{Q},  Q=P^{ \frac{ \lambda_{v_{1,p}} }{2}})$ respectively,  we have 
\begin{align}
 \E |v_{1,c}|^2  = & \frac{8  \gamma^2 \cdot \frac{ 1}{Q^2} }{ 2Q +1}  \sum_{i=1}^{Q} i^2  \leq   \frac{  8  \gamma^2 }{3}   \quad \text{for} \quad Q= P^{ \frac{ \lambda_{v_{1,c}} }{2}} \non \\
  \E |v_{1,p}|^2  =  & \frac{2  \gamma^2 \cdot \frac{ 1}{Q^2} }{ 2Q +1}  \sum_{i=1}^{Q} i^2  \leq   \frac{  2  \gamma^2 }{3}   \quad \text{for} \quad Q=P^{ \frac{ \lambda_{v_{1,p}} }{2}}    \non
\end{align} 
which implies that  
\begin{align}
 &\E |x_1|^2    \non\\ 
 =  & P^{ - \beta_{v_{1,c}}}   |h_{2 3}|^2  |h_{1 2}|^2  \E |v_{1,c}|^2   +   P^{ - \beta_{v_{1,p}}}      |h_{2 3}|^2   |h_{1 2}|^2 \E |v_{1,p}|^2         \non\\
  \leq & 16   \times \frac{ 8  \gamma^2 }{3}   + 16   \times \frac{ 2 \gamma^2 }{3} \non\\    =  &  \frac{ 160 \gamma^2 }{3}  \non\\   \leq &1     \non
\end{align} 
where $h_{ \ell k} \in (1, 2], \forall \ell, k$ and $\gamma  \in \bigl(0, \frac{1}{8\sqrt{2}}\bigr]$. Similarly, one can  check that $\E |x_2|^2 \leq 1$ and $\E |x_3|^2 \leq 1$.
Based on the above signal design,  the  received signals at the receivers~1 and~2 take the following forms (without the time index)
\begin{align}
y_{1} &=    \sqrt{P^{ 1 - \beta_{v_{1,c}}}} h_{11}    h_{23} h_{12} v_{1,c}  +    \sqrt{P^{ 1 - \beta_{v_{1,p}}}} h_{11}    h_{23} h_{12}  v_{1,p}   
\non\\& \quad  +     h_{12}h_{21} h_{13}(   \sqrt{P^{ \alpha - \beta_{v_{2,c}}}} v_{2,c}  +   \sqrt{P^{ \alpha - \beta_{u_3}}} u_3)    \non\\
&\quad  +    \sqrt{P^{ \alpha - \beta_{v_{2,p}}}}  h_{12}h_{21} h_{13} v_{2,p}   +  z_{1}    \label{eq:yvk1}  \\
y_{2} &=     \sqrt{P^{ 1 - \beta_{v_{2,c}}}} h_{22}h_{13}h_{21}  v_{2,c}    +    \sqrt{P^{ 1 - \beta_{v_{2,p}}}} h_{22}h_{13}h_{21}  v_{2,p}    \non\\& \quad+     h_{21}  h_{12} h_{23}(   \sqrt{P^{ \alpha - \beta_{v_{1,c}}}} v_{1,c}  +   \sqrt{P^{\alpha - \beta_{u_3}}} u_3)       \non\\
 & +   \sqrt{P^{ \alpha - \beta_{v_{1,p}}}}  h_{21}  h_{12}  h_{23} v_{1,p}   +  z_{2} .   \label{eq:yvk2}  
\end{align}
 Based on our signal design,  at receiver~1 the signal $v_{2,c}$ is aligned with the jamming signal $u_3$, while at receiver~2 the signal $v_{1,c}$ is aligned with the jamming signal $u_3$.

\subsubsection{Secure rate analysis} We define  the rates  $R_k$ and $R_k'$ as 
\begin{align}
R_k &\defeq   \Imu(v_k; y_k) -  \Imu ( v_k; y_{\ell} | v_{\ell} ) - \epsilon   \label{eq:Rk623J} \\  
R_k'  &\defeq  \Imu ( v_k; y_{\ell} | v_{\ell}) - \epsilon  \label{eq:Rk623bJ}  
\end{align}
for some $\epsilon >0$, $ \ell, k =1,2, \ell \neq k$. Given the above  codebook design and signal mapping, the result of \cite[Theorem~2]{XU:15}   (or \cite[Theorem~2]{LMSY:08}) implies that the rate pair $(R_1, R_2)$ defined in  \eqref{eq:Rk623J} and \eqref{eq:Rk623bJ} is achievable  and  messages $w_1$ and $w_2$ are secure, i.e.,   $\Imu(w_1; y_{2}^{\bln})  \leq  \bln \epsilon$ and $\Imu(w_2; y_{1}^{\bln})  \leq  \bln \epsilon$. 

In what follows we will focus on the regime of  $ \alpha > 1/2$ and analyze the secure rate performance of the cooperative jamming scheme. We will consider four cases:   $\frac{1}{2} < \alpha \leq \frac{2}{3}$,  $\frac{2}{3} \leq  \alpha \leq 1$, $1 \leq \alpha \leq 2$   and  $2 \leq \alpha $. 
We will first consider the  cases of $\frac{1}{2} < \alpha \leq \frac{2}{3}$ and $2 \leq \alpha $,  in which a successive decoding method  will be used in the rate analysis. Later on we will consider the rest two cases,  in which the estimation approaches of noise removal and signal separation due to rational independence will be used in the rate analysis.

\begin{table}
\caption{Parameter design for the symmetric  channel, for some $\epsilon >0$.}
\begin{center}
{\renewcommand{\arraystretch}{1.7}
\begin{tabular}{|c|c|c|c|c|}
  \hline
                     &   $\frac{1}{2} < \alpha \leq \frac{2}{3}$  &  $\frac{2}{3} \leq  \alpha \leq 1$  & $1 \leq \alpha \leq 2$   & $2 \leq \alpha $  \\
   \hline
   $\beta_{v_{1,c}}, \ \beta_{v_{2,c}}$    		&   $0$     			&   $ 0$    &    $0 $    &   $0 $  \\
    \hline
   $\beta_{u_3}$     		     		&   $0$      	&    $0$   &   0    &    0  \\
    \hline
   $\beta_{v_{1,p}}, \ \beta_{v_{2,p}}$ 			&   $\alpha$    		&    $\alpha$    &   $\infty$    &    $\infty$  \\
    \hline
   $\lambda_{v_{1,c}}, \ \lambda_{v_{2,c}}$ 		&   $2\alpha -1 - \epsilon$ 	&  $\alpha/2 - \epsilon$     &   $\alpha/2 - \epsilon$     &  $1 - \epsilon$   \\
    \hline
   $\lambda_{u_3}$ 					&    $2\alpha -1 - \epsilon$ &   $\alpha/2 - \epsilon$    &   $\alpha/2 - \epsilon$     &   $1 - \epsilon$  \\
  \hline
   $\lambda_{v_{1,p}}, \ \lambda_{v_{2,p}}$  &   $1 - \alpha - \epsilon$  	&  $1- \alpha - \epsilon$     &     $0 $  &    0  \\
    \hline
    \end{tabular}
}
\end{center}
\label{tab:para}
\end{table}

 \subsection{Rate analysis when $1/2 <  \alpha \leq  2/3$   \label{sec:CJscheme2334}}

 For the case with $1/2 <  \alpha \leq  2/3$, the parameters are designed as 
 \begin{align}
\beta_{v_{1,c}}&= \beta_{v_{2,c}} =0,        \quad    \lambda_{v_{1,c}} = \lambda_{v_{2,c}}= 2\alpha -1 - \epsilon     \label{eq:para111} \\
  \beta_{v_{1,p}}&=\beta_{v_{2,p}}= \alpha,  \quad  \lambda_{v_{1,p}}= \lambda_{v_{2,p}}  = 1 - \alpha - \epsilon  \label{eq:para333924} \\
    \beta_{u_3}& = 0  ,    \quad  \quad\quad  \quad  \quad  \lambda_{u_3}= 2\alpha -1 - \epsilon,   \label{eq:para222} 
 \end{align}
 where  $\epsilon>0$ can be set arbitrarily small.
 In this case, the transmitted signal at transmitter~$k$, $k=1, 2, 3$, is 
 \begin{align}
   x_1  = &     h_{2 3}  h_{1 2} v_{1,c} +   \sqrt{P^{ -  \alpha}}    \cdot  h_{2 3}  h_{1 2} v_{1,p}      \label{eq:xvkkk123400}  \\
   x_2  = &     h_{1 3}  h_{2 1} v_{2,c} +   \sqrt{P^{ -  \alpha}}    \cdot  h_{1 3}  h_{2 1} v_{2,p}   \label{eq:xvkkk123411}  \\
      x_3 = &    h_{1 2}  h_{21}  u_{3}.   \label{eq:xvkkk123433}  
 \end{align}
Then the received signals at the receivers~1 and~2 are given by
\begin{align}
y_{1} &=    \sqrt{P} h_{11}    h_{23} h_{12} v_{1,c}  +    \sqrt{P^{ 1 - \alpha}} h_{11}    h_{23} h_{12}  v_{1,p}   
\non\\& \quad  +    \sqrt{P^{ \alpha }} h_{12}h_{21} h_{13}  (   v_{2,c}  +  u_3) +   h_{12}h_{21} h_{13} v_{2,p}   +  z_{1}    \label{eq:yvk12334}  \\
y_{2} &=       \sqrt{P} h_{22}h_{13}h_{21}  v_{2,c}    +   \sqrt{P^{ 1 - \alpha}}h_{22}h_{13}h_{21}  v_{2,p}  
 \non\\ &\quad  +    \sqrt{P^{ \alpha }}   h_{21}  h_{12} h_{23}(    v_{1,c}  +  u_3)    +   h_{21}  h_{12}  h_{23} v_{1,p}   +  z_{2} .    \label{eq:yvk22334}   
\end{align}
Fig.~\ref{fig:allocation_a23} depicts the rate (GDoF) and power of some signals at receiver~1   when $\alpha = 2/3$.  

\begin{figure}[t!]
\centering
\includegraphics[width=7cm]{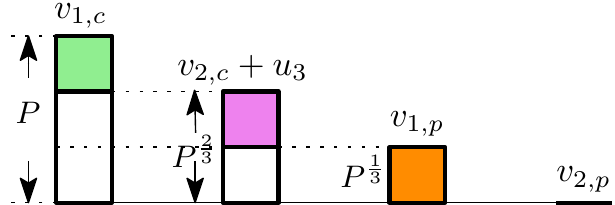}
 \vspace{-.05 in}
\caption{Rate (GDoF) and power description at receiver~1  when $\alpha =  2/3$.}
\label{fig:allocation_a23}
\end{figure}

For the proposed  cooperative jamming  scheme, the secure rate pair $(R_1, R_2)$ defined in \eqref{eq:Rk623J} and \eqref{eq:Rk623bJ} can be achieved. For $\epsilon \to 0$, this secure rate pair can be written as    
\begin{align}
R_1  & =    \Imu(v_1; y_1) - \Imu(v_1; y_2 | v_2 )      \label{eq:lboundit1} \\
R_2  & =    \Imu(v_2; y_2) - \Imu(v_2; y_1 | v_1 ) .    \label{eq:lboundit2}
\end{align}

First, we focus on the lower bound of $R_1$ expressed  in \eqref{eq:lboundit1} and seek to get a tight lower bound on $ \Imu(v_1; y_1)$.
For this case, $v_1$  is designed as \[v_1 = v_{1,c} +   \sqrt{P^{ -\alpha}}   \cdot v_{1,p}\] (see \eqref{eq:xvk}, \eqref{eq:para111} and \eqref{eq:para333924}).   
From $y_1$ expressed in \eqref{eq:yvk12334},  we can estimate $\{v_{1,c}, v_{1,p}\}$  by using a successive decoding method,  given the  design in \eqref{eq:constellationGsym1}-\eqref{eq:constellationGsym1u} and \eqref{eq:para111}-\eqref{eq:xvkkk123433}.  
From the steps in \eqref{eq:rateAWGNIC1202e}-\eqref{eq:rateAWGNIC1405e}, $\Imu(v_1; y_1)$  can be lower bounded by
 \begin{align}
  \Imu(v_1; y_1)  
  \geq  \bigl( 1 -   \text{Pr} [  \{ v_{1,c} \neq \hat{v}_{1,c} \} \cup  \{ v_{1,p} \neq \hat{v}_{1,p} \}  ] \bigr)   \cdot \Hen(v_{1})  - 1 . \label{eq:rateAWGNIC1405}     
 \end{align} 
Given that  $v_{1,c} \in    \Omega (\xi   =  \gamma_{v_{1,c}} \cdot \frac{ 1}{Q},   \   Q =  P^{ \frac{ 2\alpha -1 - \epsilon}{2}} ) $ and $v_{1,p}  \in    \Omega (\xi   =\gamma_{v_{1,p}} \cdot \frac{ 1}{Q},   \   Q = P^{ \frac{ 1 - \alpha - \epsilon}{2}} ) $,  the rates of $v_{1,c}$ and $ v_{1,p}$ are computed as follows:
  \begin{align}
  \Hen(v_{1,c}) &=  \log (2 \cdot P^{ \frac{ 2\alpha -1 - \epsilon}{2}} +1)   \label{eq:rateAWGNIC111}     \\
  \Hen (v_{1,p}) &=  \log (2 \cdot P^{ \frac{ 1 - \alpha - \epsilon}{2}} +1)  \label{eq:rateAWGNIC222}                              
 \end{align}
In this case, $\{ v_{1,p}, v_{1,c}\}$ can be reconstructed from $v_{1}$,  and vice versa. 
This fact, together with  \eqref{eq:rateAWGNIC111} and \eqref{eq:rateAWGNIC222}, gives 
  \begin{align}
 \Hen(v_{1}) 
 &= \Hen(v_{1,c})+  \Hen(v_{1,p})   \non\\
 &=   \frac{ \alpha  - 2\epsilon}{2} \log P + o(\log P) . \label{eq:rateAWGNIC333}                                  
 \end{align}
To further derive the lower bound on $ \Imu(v_1; y_1)$ from \eqref{eq:rateAWGNIC1405}, we provide an upper bound on the error probability $ \text{Pr} [  \{ v_{1,c} \neq \hat{v}_{1,c} \} \cup  \{ v_{1,p} \neq \hat{v}_{1,p} \}  ]$, described in the following lemma. 
 \begin{lemma}  \label{lm:rateerror2334}
When $1/2 < \alpha \leq 2/3$, given the signal design in \eqref{eq:constellationGsym1}-\eqref{eq:constellationGsym1u} and \eqref{eq:para111}-\eqref{eq:xvkkk123433},  the error probability of estimating  $\{v_{k,c}, v_{k,p} \}$ from $y_k$, $k=1,2$,  is
 \begin{align}
 \text{Pr} [  \{ v_{k,c} \neq \hat{v}_{k,c} \} \cup  \{ v_{k,p} \neq \hat{v}_{k,p} \}  ]  \to 0         \quad \text {as}\quad  P\to \infty .   \label{eq:error1c1p}
 \end{align}
 \end{lemma}

\begin{proof}
See Section~\ref{sec:rateerror2334}.
\end{proof}

By combining  \eqref{eq:rateAWGNIC1405},  \eqref{eq:rateAWGNIC333}  and Lemma~\ref{lm:rateerror2334},  $ \Imu(v_1; y_1)$ can be lower bounded as 
  \begin{align}
  \Imu(v_1; y_1)   &\geq    \frac{ \alpha  - 2\epsilon}{2} \log P + o(\log P) .   \label{eq:rateAWGNIC17762}  
 \end{align}

Let us now focus on  the term $\Imu(v_1; y_2 | v_2 )$ in \eqref{eq:lboundit1}, which can be upper bounded by
  \begin{align}
  &\Imu(v_1; y_2 | v_2 )     \non\\
\leq &  \Imu(v_1; y_2,  v_{1,c} + u_3  | v_2 )     \label{eq:rateAWGNIC18374}   \\
=   & \Imu(v_1;    v_{1,c} + u_3  )   +  \Imu(v_1;    h_{21}  h_{12}  h_{23} v_{1,p}  +  z_{2}   | v_2 , v_{1,c} + u_3 )     \label{eq:rateAWGNIC18735}   \\
=   & \Hen( v_{1,c} + u_3  )  - \Hen(u_3  )    \non\\&   +  \hen( h_{21}  h_{12}  h_{23} v_{1,p}  +  z_{2}   |  v_2 , v_{1,c} + u_3)   -  \hen(   z_{2} )    \label{eq:rateAWGNIC77364}    \\
\leq &   \log (4  \cdot P^{ \frac{ 2\alpha -1 - \epsilon}{2}} +1)    -   \log (2  \cdot P^{ \frac{ 2\alpha -1 - \epsilon}{2}} +1)   \non\\& +  \hen( h_{21}  h_{12}  h_{23} v_{1,p}   +  z_{2} )     -  \frac{1}{2}\log (2 \pi e)    \label{eq:rateAWGNIC1955}   \\
\leq &  1   +  \frac{1}{2}\log (2\pi e \times 65 )       -  \frac{1}{2} \log (2\pi e)    \label{eq:rateAWGNIC8356}   \\
= & \log (2\sqrt{65})    \label{eq:rateAWGNIC9562}  
 \end{align}
where  
 \eqref{eq:rateAWGNIC18735} follows from the fact that $v_1, v_2, u_3 $ are mutually independent;
\eqref{eq:rateAWGNIC77364}   stems from that $\{ v_{k,p}, v_{k,c}\}$ can be reconstructed from $v_{k}$ for $k=1,2$;
\eqref{eq:rateAWGNIC1955} holds true because   $\Hen(u_3  ) =   \log (2  \cdot P^{ \frac{ 2\alpha -1 - \epsilon}{2}} +1) $  and $\Hen( v_{1,c} + u_3  ) \leq  \log (4  \cdot P^{ \frac{ 2\alpha -1 - \epsilon}{2}} +1)$;  
\eqref{eq:rateAWGNIC8356} follows from the derivation that  $\hen( h_{21}  h_{12}  h_{23} v_{1,p}   +  z_{2} ) \leq  \frac{1}{2} \log ( 2 \pi e ( | h_{21}|^2 \cdot | h_{12}|^2 \cdot | h_{23}|^2\cdot \E |v_{1,p}|^2     +  \E|  z_{2}|^2 ))  \leq  \frac{1}{2} \log ( 2 \pi e\times 65 )$.

Finally,  by incorporating \eqref{eq:rateAWGNIC17762}  and \eqref{eq:rateAWGNIC9562} into \eqref{eq:lboundit1}, we can bound the secure rate $R_1$ as 
\begin{align}
R_1   =    \Imu(v_1; y_1) - \Imu(v_1; y_2 | v_2 )       \geq   \frac{\alpha  - 2\epsilon}{2} \log P + o(\log P).   \non 
\end{align}
Let $\epsilon \to 0$, then the secure GDoF $d_1 =  \alpha$ is achievable.   
Due to the symmetry,  $d_2 =  \alpha$ is also achievable by the proposed cooperative jamming scheme when  $1/2 <  \alpha \leq  2/3$.

 \subsection{Rate analysis when $\alpha \geq  2$   \label{sec:CJscheme2}}

  For the case with $ \alpha \geq  2$, the parameters  are designed as 
 \begin{align}
\beta_{v_{1,c}}&= \beta_{v_{2,c}} =0,        \quad    \lambda_{v_{1,c}} = \lambda_{v_{2,c}}= 1 - \epsilon     \label{eq:para11122} \\
  \beta_{v_{1,p}}&=\beta_{v_{2,p}}= \infty,  \  \   \lambda_{v_{1,p}}= \lambda_{v_{2,p}}  = 0  \label{eq:para33392422} \\
    \beta_{u_3}& = 0 ,\quad  \quad\quad  \quad  \quad  \lambda_{u_3}= 1 - \epsilon.   \label{eq:para22222} 
 \end{align}
In this case, the transmitted signals are designed as
  \begin{align}
   x_1  =  \!    h_{2 3}  h_{1 2} v_{1,c} , \     
   x_2  =    \!  h_{1 3}  h_{2 1} v_{2,c} , \   
      x_3 =   \!  h_{1 2}  h_{21}  u_{3}.   \label{eq:xvkkk12343311}  
 \end{align}
Then the received signals at the receivers~1 and~2 become
\begin{align}
y_{1} &=    \sqrt{P} h_{11}    h_{23} h_{12} v_{1,c}   +    \sqrt{P^{ \alpha }} h_{12}h_{21} h_{13}  (   v_{2,c}  +  u_3)    +  z_{1}    \label{eq:yvk122}  \\
y_{2} &=       \sqrt{P} h_{22}h_{13}h_{21}  v_{2,c}     +    \sqrt{P^{ \alpha }}   h_{21}  h_{12} h_{23}(    v_{1,c}  +  u_3)     +  z_{2} .    \label{eq:yvk222}   
\end{align}
 Figure~\ref{fig:allocation_b2} depicts the rate and power of some signals at receiver~1 when $\alpha = 2$.

\begin{figure}[t!]
\centering
\includegraphics[width=3.5cm]{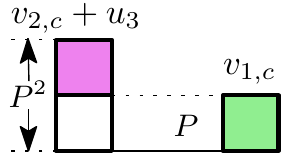}
 \vspace{-.05 in}
 \caption{Rate (GDoF) and power description at receiver~1  when $\alpha =  2$.}
\label{fig:allocation_b2}
\end{figure}

From \eqref{eq:Rk623J} and \eqref{eq:Rk623bJ},  the  secure rates  $R_1 = \Imu(v_1; y_1) - \Imu(v_1; y_2 | v_2 ) $ and  $R_2  =    \Imu(v_2; y_2) - \Imu(v_2; y_1 | v_1 ) $ can be achieved by the proposed scheme.
For this case, $v_k$  is designed as  \[v_k =  v_{k,c} , \quad k=1,2. \] 
From $y_k$ expressed in \eqref{eq:yvk122} and \eqref{eq:yvk222},  we can estimate $v_{k,c}$   by using a successive decoding method, for $k=1,2$. 
Lemma~\ref{lm:rateerror322}  provides a result on the error probability for this estimation.

 \begin{lemma}  \label{lm:rateerror322}
When $ \alpha \geq 2$, given the signal design in \eqref{eq:constellationGsym1}-\eqref{eq:constellationGsym1u} and \eqref{eq:para11122}-\eqref{eq:xvkkk12343311},  the error probability of  estimating  $v_{k,c}$ from $y_k$, $k=1,2$,  is
 \begin{align}
 \text{Pr} [ v_{k,c} \neq \hat{v}_{k,c}  ]  \to 0         \quad \text {as}\quad  P\to \infty  .  \label{eq:error1c1p322}
 \end{align}
 \end{lemma}

\begin{proof}
See Appendix~\ref{sec:rateerror322}.
\end{proof}

Considering the lower bound on the secure rate $R_1$, 
$\Imu(v_1; y_1)$ can be lower bounded by
  \begin{align}
  \Imu(v_1; y_1)   \geq      \bigl( 1 -   \text{Pr} [   v_{1,c} \neq \hat{v}_{1,c}  ] \bigr)   \cdot \Hen(v_{1})  - 1  \label{eq:rateAWGNIC1405322}     
      \end{align}
 by following the steps in \eqref{eq:rateAWGNIC1202e}-\eqref{eq:rateAWGNIC1405e}. 
In this case with   $v_1 =  v_{1,c}$ and  $v_{1,c} \in \Omega ( \xi =  \gamma_{v_{1,c}} \cdot \frac{ 1}{Q} ,   \   Q =  P^{ \frac{1 - \epsilon }{2}} )$,  the rate of $v_{1} $ is given by  
  \begin{align}
 \Hen(v_{1}) =  \Hen(v_{1,c}) &=  \log (2 \cdot P^{ \frac{1  - \epsilon }{2}} +1) .  \label{eq:rateAWGNIC1c322}     
 \end{align}

By combining  \eqref{eq:rateAWGNIC1c322} and Lemma~\ref{lm:rateerror322}, the lower bound of $\Imu(v_1; y_1)$ is given as
\begin{align}
 \Imu(v_1; y_1)  \geq   \frac{ 1  - \epsilon}{2} \log P + o(\log P) .   \label{eq:rateAWGNIC17762322}  
\end{align}

By following the steps \eqref{eq:rateAWGNIC18374}-\eqref{eq:rateAWGNIC9562},  $\Imu(v_1; y_2 | v_2 )$ can be bounded by 
 \begin{align}
  &\Imu(v_1; y_2 | v_2 )     \non\\
\leq &  \Imu(v_1; y_2,  v_{1,c} + u_3  | v_2 )    \non \\
=  &    \Imu(v_1;    v_{1,c} + u_3  |  v_2 )   + \Imu(v_1;    y_2  |  v_2, v_{1,c} + u_3 )      \non  \\
=   & \Imu(v_1;    v_{1,c} + u_3  )   + \underbrace{  \Imu(v_1;     z_{2}   | v_2 , v_{1,c} + u_3 ) }_{=0}   \non\\
=   & \Hen( v_{1,c} + u_3  )  - \Hen(u_3  )        \non\\
\leq &   \log (4  \cdot P^{ \frac{ 1 - \epsilon}{2}} +1)    -   \log (2  \cdot P^{ \frac{ 1 - \epsilon}{2}} +1)     \non\\ 
\leq &  1   .    \label{eq:rateAWGNIC9562322}  
 \end{align}

Finally,  with the results in \eqref{eq:rateAWGNIC17762322}  and \eqref{eq:rateAWGNIC9562322}, we can bound the secure rate $R_1$ as 
\begin{align}
R_1   =    \Imu(v_1; y_1) - \Imu(v_1; y_2 | v_2 )    \geq  \frac{ 1  - \epsilon}{2} \log P + o(\log P).  \non 
\end{align}
Let $\epsilon \to 0$, then the  secure GDoF $d_1 =  1$ is achievable.   
Due to the symmetry,  $d_2 =  1$ is also achievable by the proposed cooperative jamming scheme when  $ \alpha \geq 2$.

\subsection{Rate analysis when $2/3 \leq  \alpha \leq  1$   \label{sec:CJscheme231}}

The rate analysis for this case  is different from that for the previous two cases.  In the previous two cases, a successive decoding method  is used in the rate analysis. 
In this case, we will use the estimation approaches of noise removal and signal separation due to rational independence that will be discussed later on.

 For this case with $2/3 \leq  \alpha \leq  1$,  the parameters are designed as 
  \begin{align}
\beta_{v_{1,c}}&= \beta_{v_{2,c}} =0,        \quad    \lambda_{v_{1,c}} = \lambda_{v_{2,c}}= \alpha/2 - \epsilon     \label{eq:para34111} \\
  \beta_{v_{1,p}}&=\beta_{v_{2,p}}= \alpha,  \   \  \lambda_{v_{1,p}}= \lambda_{v_{2,p}}  = 1 - \alpha - \epsilon  \label{eq:para854673} \\
    \beta_{u_3}& = 0 ,\quad  \quad\quad \quad  \   \  \lambda_{u_3}= \alpha/2 - \epsilon,   \label{eq:para746376} 
 \end{align}
 where  $\epsilon>0$ can be set arbitrarily small.
In this case, the transmitted signals are designed as
 \begin{align}
   x_1  = &     h_{2 3}  h_{1 2} v_{1,c} +   \sqrt{P^{ -  \alpha}}    \cdot  h_{2 3}  h_{1 2} v_{1,p}      \label{eq:xvkkk23100}  \\
   x_2  = &     h_{1 3}  h_{2 1} v_{2,c} +   \sqrt{P^{ -  \alpha}}    \cdot  h_{1 3}  h_{2 1} v_{2,p}   \label{eq:xvkkk23111}  \\
      x_3 = &     h_{1 2}  h_{21}  u_{3}.   \label{eq:xvkkk1341}  
 \end{align}
Then the received signals at the receivers~1 and~2  take the same forms as in   \eqref{eq:yvk12334} and \eqref{eq:yvk22334}.
Fig.~\ref{fig:allocation_c34} describes the rate and power of some signals at receivers~1  and~2  when $\alpha = 3/4$.

In this proposed scheme,     the  secure rates $R_1  =    \Imu(v_1; y_1) - \Imu(v_1; y_2 | v_2 ) $ and $R_2  =    \Imu(v_2; y_2) - \Imu(v_2; y_1 | v_1 ) $ are achievable  (see  \eqref{eq:Rk623J} and \eqref{eq:Rk623bJ}).
Let us bound the secure rate $R_1$ first. 
By following the steps in \eqref{eq:rateAWGNIC1202e}-\eqref{eq:rateAWGNIC1405e},  $\Imu(v_1; y_1)$ can be lower bounded by
  \begin{align}
  \Imu(v_1; y_1)    \geq \bigl( 1 -   \text{Pr} [  \{ v_{1,c} \neq \hat{v}_{1,c} \} \cup  \{ v_{1,p} \neq \hat{v}_{1,p} \}  ] \bigr)   \cdot \Hen(v_{1})  - 1 . \label{eq:rateAWGNIC1405341}    
 \end{align}

For this case, the rates of $v_{1,c}$, $v_{1,p}$ and $v_{1} =  v_{1,c} +   \sqrt{P^{ - \alpha}}    \cdot  v_{1,p} $ are  computed as
  \begin{align}
  \Hen(v_{1,c}) &=  \log (2 \cdot P^{ \frac{ \alpha/2  - \epsilon}{2}} +1)   \label{eq:rateAWGNIC111341}     \\
  \Hen (v_{1,p}) &=  \log (2 \cdot P^{ \frac{ 1 - \alpha - \epsilon}{2}} +1)   \label{eq:rateAWGNIC222341}   \\
  \Hen(v_{1})
& =    \frac{1- \alpha/2  - 2\epsilon}{2} \log P + o(\log P) . \label{eq:rateAWGNIC333341}                                 
 \end{align}
 To further derive the lower bound on $ \Imu(v_1; y_1)$ from \eqref{eq:rateAWGNIC1405341}, we provide an upper bound on the error probability $ \text{Pr} [  \{ v_{1,c} \neq \hat{v}_{1,c} \} \cup  \{ v_{1,p} \neq \hat{v}_{1,p} \}  ]$, described in the following lemma. 
\begin{lemma}  \label{lm:rateerror341}
 When $2/3 \leq \alpha \leq 1$, given the signal design in \eqref{eq:constellationGsym1}-\eqref{eq:constellationGsym1u} and \eqref{eq:para34111}-\eqref{eq:xvkkk1341}, then for almost all the channel coefficients  $\{h_{k\ell}\} \in (1, 2]^{2\times 3}$,  the error probability of estimating  $\{v_{k,c}, v_{k,p} \}$ from $y_k$, $k=1,2$,  is
 \begin{align}
 \text{Pr} [  \{ v_{k,c} \neq \hat{v}_{k,c} \} \cup  \{ v_{k,p} \neq \hat{v}_{k,p} \}  ]  \to 0         \quad \text {as}\quad  P\to \infty.    \label{eq:error1c1p341}
 \end{align}
 \end{lemma}
 \begin{proof}
In this proof  we use the approaches of noise removal and signal separation. The full details are described in Section~\ref{sec:rateerror341}.
  \end{proof}

By combining \eqref{eq:rateAWGNIC1405341}, \eqref{eq:rateAWGNIC333341} and Lemma~\ref{lm:rateerror341},   $\Imu(v_1; y_1)$ can be lower bounded by
  \begin{align}
  \quad \Imu(v_1; y_1)  \geq   \frac{1- \alpha/2  - 2\epsilon}{2} \log P + o(\log P)      \label{eq:rateAWGNIC17762341}  
 \end{align}
 for  almost all channel coefficients  $\{h_{k\ell}\} \in (1, 2]^{2\times 3}$. 
By following the steps in \eqref{eq:rateAWGNIC18374}-\eqref{eq:rateAWGNIC77364},  $\Imu(v_1; y_2 | v_2 )$ can be bounded  as
   \begin{align}
  \Imu(v_1; y_2 | v_2 )     \leq  & \log (2\sqrt{65})  .  \label{eq:rateAWGNIC9562341}  
 \end{align}
Finally,  with the results in \eqref{eq:rateAWGNIC17762341}  and \eqref{eq:rateAWGNIC9562341}, the secure rate $R_1$ is lower bounded by 
\begin{align}
R_1   =    \Imu(v_1; y_1)\! -\! \Imu(v_1; y_2 | v_2 )   \!\geq\!   \frac{1- \alpha/2 - 2\epsilon}{2} \log P \!+\! o(\log P)   \non 
\end{align}
which implies the following secure GDoF $d_1 =  1- \alpha/2$, as well as $d_2 =  1- \alpha/2$  due to the symmetry,  for  almost all  the channel coefficients  $\{h_{k\ell}\} \in (1, 2]^{2\times 3}$, in this  case with  $2/3 \leq \alpha \leq  1$.

\subsection{Rate analysis when $1 \leq  \alpha \leq 2$   \label{sec:CJscheme12}}

 The rate analysis for this case  also uses  the approaches of noise removal and signal separation.  For this case,  the parameters are designed as 
 \begin{align}
\beta_{v_{1,c}}&= \beta_{v_{2,c}} =0,        \quad    \lambda_{v_{1,c}} = \lambda_{v_{2,c}}= \alpha/2 - \epsilon     \label{eq:para13211} \\
  \beta_{v_{1,p}}&=\beta_{v_{2,p}}= \infty,  \  \   \lambda_{v_{1,p}}= \lambda_{v_{2,p}}  = 0  \label{eq:para2457295} \\
    \beta_{u_3}& = 0 ,\quad  \quad\quad  \quad  \quad  \lambda_{u_3}= \alpha/2 - \epsilon.   \label{eq:para13222} 
 \end{align}
In this case, the transmitted signals are designed as  
  \begin{align}
   x_1  =  \!    h_{2 3}  h_{1 2} v_{1,c} ,   \ 
   x_2  =  \!    h_{1 3}  h_{2 1} v_{2,c},   \  
      x_3 =   \!  h_{1 2}  h_{21}  u_{3}.   \label{eq:xvkkk1132}  
 \end{align}
Then the received signals at the receivers take the forms  as in \eqref{eq:yvk122} and \eqref{eq:yvk222}.

\begin{figure}[t!]
\centering
\includegraphics[width=4cm]{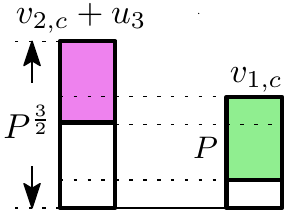}
 \vspace{-.05 in}
  \caption{Rate  and power description at receiver~1    when $\alpha =  3/2$.}
\label{fig:allocation_d32}
\end{figure}

In the following we will bound the  secure rates expressed in \eqref{eq:Rk623J} and \eqref{eq:Rk623bJ}. 
For this case, $v_k$ is designed as $v_k =    v_{k,c} ,  k=1,2$.   
 We can estimate $v_{k,c}$  from $y_k$ by using the approaches of noise removal and signal separation, $k=1,2$.   Lemma~\ref{lm:rateerror132}  presents a result on the error probability for this estimation.

  \begin{lemma}  \label{lm:rateerror132}
 When $1 \leq \alpha \leq 2$, given the signal design in \eqref{eq:constellationGsym1}-\eqref{eq:constellationGsym1u} and \eqref{eq:para13211}-\eqref{eq:xvkkk1132}, then for almost all channel coefficients  $\{h_{k\ell}\} \in (1, 2]^{2\times 3}$,  the error probability of estimating  $v_{k,c}$ from $y_k$, $k=1,2$,  is
 \begin{align}
 \text{Pr} [  v_{k,c} \neq \hat{v}_{k,c}  ]  \to 0         \quad \text {as}\quad  P\to \infty.    \label{eq:error1c1p132}
 \end{align}
 \end{lemma}
 \begin{proof}
In this proof  we use the approaches of noise removal and signal separation. The full details are described in Appendix~\ref{sec:rateerror132}.
  \end{proof}

In the following we will show that  by using Lemma~\ref{lm:rateerror132}, the secure GDoF pair $(d_1 =  \alpha/2 ,  d_2 = \alpha/2  )$ can be achieved  for  almost all the channel coefficients  $\{h_{k\ell}\} \in (1, 2]^{2\times 3}$ by the proposed cooperative jamming scheme. 
In this case with  $v_{1,c} \in \Omega ( \xi = 2 \gamma \cdot \frac{ 1}{Q} ,   \   Q =  P^{ \frac{\alpha/2 - \epsilon }{2}} )$ and  $v_1 =   v_{1,c}$,  the rate of $v_{1} $ is computed as   
  \begin{align}
 \Hen(v_{1}) =  \Hen(v_{1,c}) &=  \log (2 \cdot P^{ \frac{\alpha/2 - \epsilon }{2}} +1) .  \label{eq:rateAWGNIC1c132}     
 \end{align}
and then, $\Imu(v_1; y_1)$ can be  bounded by
  \begin{align}
  \Imu(v_1; y_1)    &\geq  \bigl( 1 -   \text{Pr} [   v_{1,c} \neq \hat{v}_{1,c}  ] \bigr)   \cdot \Hen(v_{1})  - 1  \label{eq:rateAWGNIC1405322009}       \\
        & =  \frac{ \alpha/2 - \epsilon}{2} \log P + o(\log P)    \label{eq:rateAWGNIC17762132}  
 \end{align}
 for  almost all the channel coefficients  $\{h_{k\ell}\} \in (1, 2]^{2\times 3}$, where  
 \eqref{eq:rateAWGNIC1405322009} follows from  the steps in \eqref{eq:rateAWGNIC1202e}-\eqref{eq:rateAWGNIC1405e};
  and \eqref{eq:rateAWGNIC17762132}  results from  \eqref{eq:rateAWGNIC1c132} and Lemma~\ref{lm:rateerror132}.
By following the steps related to \eqref{eq:rateAWGNIC9562322},  $\Imu(v_1; y_2 | v_2 )$  can be bounded as
  \begin{align}
\Imu(v_1; y_2 | v_2 )   \leq   1 .     \label{eq:rateAWGNIC9562132}  
 \end{align}
The final step is to  incorporate \eqref{eq:rateAWGNIC17762132}  and \eqref{eq:rateAWGNIC9562132} into \eqref{eq:Rk623J}. It then gives the lower bound on $R_1$
\begin{align}
R_1  & =    \Imu(v_1; y_1) - \Imu(v_1; y_2 | v_2 )      \non \\
& \geq  \frac{ \alpha/2  - \epsilon}{2} \log P + o(\log P)  \label{eq:lbounditfinal1132}
\end{align}
which implies the following secure GDoF $d_1 =  \alpha/2$, as well as $d_2 =  \alpha/2$  due to the symmetry,  for  almost all  the channel coefficients  $\{h_{k\ell}\} \in (1, 2]^{2\times 3}$, in this  case with  $1 \leq \alpha \leq  2$.

\section{Proof of Lemma~\ref{lm:NMb2} }   \label{sec:lm:NMb2}

In this section, we will prove Lemma~\ref{lm:NMb2}.  
Let $\beta \in (0,1]$, $\tau  \in \Zc^+$ and $ \tau >1$, $A_0, A_1, Q_0$ and $Q_1  \in \Zc^+$. Define the event
\[  B(q_0, q_1)  \defeq \{ (g_0, g_1)  \in (1, \tau]^2 :  |  A_0 g_0 q_0 + A_1 g_1 q_1  | < \beta \}    \]
and set 
\begin{align*}
B  \defeq   \bigcup_{\substack{ q_0, q_1 \in \Zc:  \\  (q_0, q_1) \neq  0,  \\  |q_k| \leq Q_k  \ \forall k }}  B(q_0, q_1) . 
\end{align*}
For $(q_0, q_1 ) \in  \{  (q'_0, q'_1):  (q'_0, q'_1) \neq  0, q'_0, q'_1 \in  \Zc,  |q'_0| \leq Q_0,   |q'_1| \leq Q_1  \}$, we will consider the following three cases:  $(q_0 \not = 0, q_1 \not = 0)$,  $(q_0 \not = 0, q_1 =0)$, and $(q_0 = 0, q_1 \not =0)$.

\begin{figure}
\centering
\includegraphics[width=4cm]{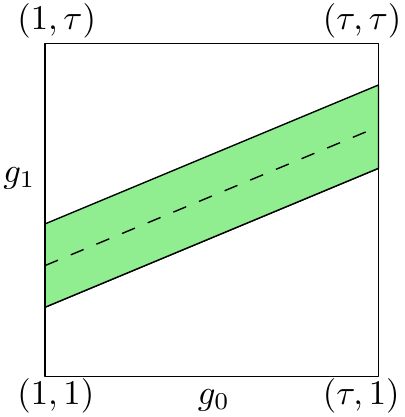}
\caption{The illustration of the set $B(q_0, q_1) \in (1, \tau]^2$ for the case of $(q_0 \not = 0, q_1 \not = 0)$. In this case, there is one strip with slope $- A_0 q_0/ (A_1 q_1)$ and width $ 2\beta/(A_1|q_1|)$.}
\label{fig:Strip}
\end{figure} 

Let us consider  the  case with $(q_0 \not = 0, q_1 \not = 0)$ first.  
In this case, assuming that $A_0 |q_0| \geq \tau A_1 |q_1|  + 1$,  then it gives 
\begin{align}
|A_0 g_0 q_0 + A_1g_1 q_1| &\geq A_0 g_0 |q_0| - A_1g_1 |q_1| \non \\
& \geq A_0 |q_0| - \tau A_1 |q_1|  \non \\
& \geq 1 \non \\
& \geq \beta
\end{align}
which contradicts the event  $|  A_0 g_0 q_0 + A_1 g_1 q_1  | < \beta$  defined in  $B(q_0, q_1) $.  Therefore,  without loss of generality we  will consider 
\begin{align}
A_0 |q_0| \leq \tau A_1 |q_1|   \label{lemmacondition1}
\end{align}
  for the  case with $(q_0 \not = 0, q_1 \not = 0)$.

 For the first case with $(q_0 \not = 0, q_1 \not = 0)$ ,  as shown in  Fig.~\ref{fig:Strip},   the set $B(q_0, q_1)$ has one strip with slope $- A_0 q_0/ (A_1 q_1)$ and width $ 2\beta/(A_1|q_1|)$.
The area of this set  is upper bounded by
\begin{align}
\Lc(B(q_0, q_1)) \leq (\tau -1) \cdot \frac{2\beta}{A_1|q_1|} = \frac{2 (\tau -1)\beta}{A_1|q_1|}. \label{lemmaarea1}
\end{align}

For  the second case with $(q_0 \not = 0, q_1 =0)$,  it holds true that $|A_0 g_0 q_0 + A_1 g_1 q_1  | = | A_0 g_0 q_0| \geq 1 \geq \beta$ for any $(g_0, g_1)  \in (1, \tau]^2 $, which implies that 
\begin{align}
\Lc (B(q_0, 0)) = 0  \label{lemmaq10}
\end{align}
based on the definition of  $B(q_0, q_1)$. 
Similarly, for the last case with  $(q_0 = 0, q_1 \not =0)$, it holds true that
\begin{align}
\Lc (B(0, q_1)) = 0.  \label{lemmaq00}
\end{align}

 By combining \eqref{lemmacondition1}-\eqref{lemmaq00},  $\Lc(B)$ can be upper bounded by
 \begin{align}
 \Lc(B) 
&= \Lc \big(    \bigcup_{\substack{ q_0, q_1 \in \Zc:  \\  (q_0, q_1) \neq  0,  \\  |q_k| \leq Q_k  \ \forall k }}   \quad B(q_0, q_1) \big) \non \\ 
& \leq \sum_{\substack{ q_0, q_1 \in \Zc:  \\  (q_0, q_1) \neq  0,  \\  |q_k| \leq Q_k  \ \forall k }}  \quad \Lc (B(q_0, q_1) ) \non\\  
& =  \sum_{q_1 \in  \Zc \setminus \{0\}: |q_1| \leq Q_1} \Lc (B(0, q_1))   +  \sum_{q_0 \in   \Zc \setminus \{0\}: |q_0| \leq Q_0} \Lc (B(q_0, 0)) \non \\ 
&\quad +  \sum_{\substack{ q_1 \in  \Zc \setminus \{0\}: \\ |q_1| \leq Q_1} } \sum_{\substack{ q_0 \in  \Zc \setminus \{0\}: \\  |q_0| \leq Q_0 \\ A_0 |q_0| \leq \tau  A_1|q_1|}} \ \Lc (B(q_0, q_1))  \non \\
& =   \sum_{\substack{ q_1 \in  \Zc \setminus \{0\}: \\ |q_1| \leq Q_1} } \sum_{\substack{ q_0 \in  \Zc \setminus \{0\}: \\  |q_0| \leq Q_0 \\ A_0 |q_0| \leq \tau  A_1|q_1|}} \ \Lc (B(q_0, q_1))  \label{lemmabound1}
\end{align}
where \eqref{lemmabound1} results from \eqref{lemmaq10} and  \eqref{lemmaq00}.
Note that 
\begin{align}
& \bigl| \{q_0 \in  \Zc \setminus \{0\}:  |q_0| \leq Q_0, A_0 |q_0| \leq \tau  A_1|q_1| \} \bigr|  \non \\ 
& \leq 2 \min \Bigl\{ Q_0, \frac{\tau  A_1|q_1|}{A_0}\Bigr\} .   \label{lemmabound8266}
\end{align}
With this,  we can bound the  term in \eqref{lemmabound1} as
\begin{align} 
 &\sum_{\substack{ q_1 \in  \Zc \setminus \{0\}: \\ |q_1| \leq Q_1} } \sum_{\substack{ q_0 \in  \Zc \setminus \{0\}: \\  |q_0| \leq Q_0 \\ A_0 |q_0| \leq \tau  A_1|q_1|}} \ \Lc (B(q_0, q_1))  \non \\ 
&\leq 2 Q_1 \cdot 2 \min \Bigl\{Q_0, \frac{\tau  A_1|q_1|}{A_0} \Bigr\} \cdot \frac{2 (\tau -1)\beta}{A_1|q_1|}  \label{lemmabound28352} \\ 
& = 8 (\tau -1)  \beta  \min \Bigl\{ \frac{Q_1 Q_0}{A_1|q_1|}, \frac{Q_1 \tau}{A_0}\Bigr\}  \non  \\ 
& \leq 8 (\tau -1)   \beta  \min \Bigl\{\frac{ Q_1 Q_0}{A_1}, \frac{Q_1 \tau}{A_0}\Bigr\}   \label{lemmabound11}
\end{align}
where \eqref{lemmabound28352} follows from  \eqref{lemmaarea1} and \eqref{lemmabound8266}.
Therefore, \eqref{lemmabound1} can be further upper bounded by
\begin{align}
\Lc (B) \leq 8 (\tau -1)   \beta  \min \Bigl\{ \frac{ Q_1 Q_0 }{A_1}, \frac{Q_1 \tau}{A_0}\Bigr\} . \label{lemmaboundfinal1}
\end{align}

Due to symmetry, by interchanging the roles of $A_0$ and $A_1$, and interchanging the roles of $Q_0$ and $Q_1$, $\Lc (B)$ can also be upper bounded by 
\begin{align}
\Lc(B) \leq 8 (\tau -1)  \beta  \min \Bigl\{ \frac{Q_0 Q_1}{ A_0}, \frac{Q_0 \tau}{A_1} \Bigr\}.
\label{lemmaboundfinal2}
\end{align}
By combining the results in \eqref{lemmaboundfinal1} and \eqref{lemmaboundfinal2}, we finally bound $\Lc(B) $ as 
\begin{align}
\Lc (B )  \leq   8 (\tau -1)  \beta \min\Bigl\{ \frac{ Q_1 Q_0}{A_1},  \frac{Q_0 Q_1}{A_0},   \frac{Q_0 \tau}{A_1},   \frac{Q_1 \tau}{A_0}\Bigr\}.
\end{align}

\section{Proof of Lemma~\ref{lm:rateerror2334}  \label{sec:rateerror2334} }

This section provides the proof of Lemma~\ref{lm:rateerror2334}. In this proof we will use \cite[Lemma~1]{ChenIC:18}  described below.
\begin{lemma}  \label{lm:AWGNic} \cite[Lemma~1]{ChenIC:18}
Consider the channel model $ y=  \sqrt{P^{\alpha_1}}  h x + \sqrt{P^{\alpha_2}} g + z$, 
where $x \in \Omega (\xi,  Q)$  is the random variable,  $z\sim \mathcal{N}(0, \sigma^2)$, and $g  \in  \Sc_{g}$ is a discrete random variable such that \[|g | \leq  g_{\max}, \quad \forall g \in  \Sc_{g}\] for a given set $\Sc_{g} \subset \Rc$.  In the above model $g_{\max}$, $\sigma$ and $h$  are positive and finite constants  independent of $P$, and  $\alpha_1$ and $\alpha_2$ are  two positive parameters  such that  $\alpha_1 - \alpha_2 >0$. By setting  $Q$ and $\xi$ such that 
\begin{align}
Q =  \frac{P^{\frac{\bar{\alpha}}{2}} \cdot h \gamma }{2 g_{\max} },    \quad  \quad    \xi =  \gamma \cdot \frac{ 1}{Q},   \quad   \quad   \forall \bar{\alpha} \in (0, \alpha_1 -\alpha_2)  \label{eq:settingAWGNIC}                           
\end{align}
where $\gamma > 0$ is a finite constant independent of $P$, then 
the probability of error for decoding a symbol $x$ from $y$ is \[ \text{Pr} (e) \to 0  \quad  \text{as} \quad   P\to \infty.\] 
\end{lemma}
\vspace{10pt}
We will provide the proof of Lemma~\ref{lm:rateerror2334} by focusing on the case  with $k=1$, as the case with $k=2$ can be proved in a similar way. 
In this proof, we estimate $v_{1,c}$,  $v_{2,c} + u_3$ and $v_{1,p}$ from $y_1$ expressed in \eqref{eq:yvk12334} by using a successive decoding method.
At first, $v_{1,c} \in \Omega ( \xi =  \gamma_{v_{1,c}} \cdot \frac{ 1}{Q} ,   \   Q =  P^{ \frac{ 2\alpha -1 - \epsilon }{2}} )$  will be estimated  by treating the other signals as noise.   Let us  rewrite   $y_1$ in \eqref{eq:yvk12334}  as 
\begin{align}
 y_1  =    \sqrt{P} h_{11}    h_{23} h_{12} v_{1,c}       +  \sqrt{P^{ \alpha}} g   +  z_{1}    \label{eq:yvk101}
\end{align}
where 
  \begin{align}
 g    \defeq       h_{12}h_{21} h_{13}  (   v_{2,c}  +  u_3)   +  \sqrt{P^{ 1 - 2\alpha}} h_{11}    h_{23} h_{12}  v_{1,p}  +  \sqrt{P^{  - \alpha}}  h_{12}h_{21} h_{13} v_{2,p} . \non
 \end{align}
For this case with  $1/2 < \alpha \leq 2/3$,  it is true that $ |g | \leq  4 \sqrt{2} $ for any realizations of $g$.
At this point, by using the result of Lemma~\ref{lm:AWGNic} it implies that   the probability of error for estimating $v_{1,c}$ from $y_1$  is 
 \begin{align}
\text{Pr} [v_{1,c} \neq \hat{v}_{1,c}] \to 0, \quad \text {as}\quad  P\to \infty.   \label{eq:error776}
\end{align}
The decoded $v_{1,c}$ can be removed from $y_1$, which can allow us to estimate $v_{2,c} + u_3 $  from  the following observation 
 \begin{align}
 y_1-  \sqrt{P} h_{11}    h_{23} h_{12} v_{1,c}    =       \sqrt{P^{ \alpha}}  h_{12}h_{21} h_{13}  (   v_{2,c}  +  u_3)    +   \sqrt{P^{ 1- \alpha }}  g'  +  z_{1}       \label{eq:yvk102}
\end{align}
 where   $ g'  \defeq   h_{11}    h_{23} h_{12}  v_{1,p}   +  \sqrt{P^{   \alpha -1 }}  h_{12}h_{21} h_{13} v_{2,p}$ and  $v_{2,c} + u_3 \in   2\cdot \Omega (\xi   =  \gamma_{v_{2,c}} \cdot \frac{ 1}{Q},   \   Q =  P^{ \frac{ 2\alpha -1 - \epsilon }{2}} )$  for $2\cdot \Omega (\xi,  Q)  \defeq   \{ \xi \cdot a :   \    a \in  \Zc  \cap [-2Q,   2Q]   \}$. It is true that $ |g' | \leq  2  \sqrt{2}$ for any realizations of $g'$.
Define  $\hat{s}_{vu}$ as the estimate of $s_{vu} \defeq v_{2,c} + u_3$. 
From the result of Lemma~\ref{lm:AWGNic}, it reveals that 
\begin{align}
\text{Pr} [s_{vu}  \neq  \hat{s}_{vu} |  v_{1,c} = \hat{v}_{1,c}] \to 0, \quad \text {as}\quad  P\to \infty. \label{eq:errorsvu}
\end{align}
which,  together with \eqref{eq:error776}, gives   
\begin{align}
 \text{Pr} [s_{vu}  \neq  \hat{s}_{vu} ] 
  \leq  \text{Pr} [s_{vu}  \neq  \hat{s}_{vu} |  v_{1,c} = \hat{v}_{1,c}] +  \text{Pr} [v_{1,c} \neq \hat{v}_{1,c}]        \to 0   \  \text {as } \ P \to \infty.    \label{eq:error887}
\end{align}

Similarly, by removing the decoded $v_{2,c} + u_3$ from $y_1$,   $v_{1,p} \in    \Omega (\xi   =\gamma_{v_{1,p}} \cdot \frac{ 1}{Q},   \   Q = P^{ \frac{ 1 - \alpha - \epsilon}{2}} )$ can be estimated  with vanishing error probability, i.e., 
  \begin{align}
 \text{Pr} [ v_{1,p} \neq  \hat{v}_{1,p}] &   \to 0    \quad \text {as }\quad  P \to \infty.   \label{eq:error2950}
\end{align}

By combining  \eqref{eq:error776} and \eqref{eq:error2950}, it gives that 
the  probability of error for estimating $\{v_{1,c}, v_{1,p} \}$ from $y_1$ is
\begin{align}
\text{Pr} [  \{ v_{1,c} \neq \hat{v}_{1,c} \} \cup  \{ v_{1,p} \neq \hat{v}_{1,p} \}  ]    \to   0         \quad \text {as}\quad  P\to \infty   \non 
\end{align}
which completes the proof for the case of $k=1$. 
Due to the symmetry,  the case of $k=2$ can be proved in a similar way.

\section{Proof of Lemma~\ref{lm:rateerror341}  \label{sec:rateerror341} }

The proof of Lemma~\ref{lm:rateerror341} is provided in this section.  
In this proof,  $v_{1,c}$,  $v_{2,c} + u_3$ and $v_{1,p}$ will be estimated from $y_1$  (see \eqref{eq:yvk12334}) by using the approaches of noise removal and signal separation. 
At first,  the following two symbols
  \begin{align}
v_{1,c} &\in \Omega ( \xi =  \gamma_{v_{1,c}} \cdot \frac{ 1}{Q} ,   \   Q =  P^{ \frac{ \alpha/2 - \epsilon }{2}} )    \non \\
v_{2,c}  +  u_{3}  &\in 2 \cdot \Omega ( \xi =  \gamma_{v_{2,c}} \cdot \frac{ 1}{Q} ,   \   Q =  P^{ \frac{ \alpha/2 - \epsilon }{2}} ) \non
 \end{align}
 will be estimated simultaneously from $y_1$  by treating the other signals as noise. 
 Let us rewrite  $y_1$ in \eqref{eq:yvk12334}  as 
\begin{align}
y_{1} & =  \sqrt{P} h_{11}    h_{23} h_{12} v_{1,c}  +  \sqrt{P^{ \alpha }} h_{12}h_{21} h_{13}  (   v_{2,c}  +  u_3)  + \tilde{z}_{1} \non\\
& = P^{{ \frac{ \alpha/2 + \epsilon }{2}}} \cdot  2\gamma \cdot  (  \sqrt{P^{ 1- \alpha }}  g_0 q_0 + g_1 q_1 )  + \tilde{z}_{1} \label{eq:rewritteny231}
\end{align}
where   $\tilde{z}_{1}   \defeq     \sqrt{P^{ 1 - \alpha}} h_{11}    h_{23} h_{12}  v_{1,p}  +    h_{12}h_{21} h_{13} v_{2,p}   +  z_{1}$ and 
 \[ g_0\defeq  h_{11}    h_{23} h_{12},     \quad  g_1\defeq   h_{12}h_{21} h_{13}  \] 
    \[ q_0  \defeq  \frac{Q_{\max}}{2\gamma} \cdot   v_{1,c}   ,    \quad  q_1  \defeq  \frac{Q_{\max}}{2\gamma} \cdot   (   v_{2,c}  +  u_{3}),  \quad Q_{\max} \defeq P^{ \frac{ \alpha/2 - \epsilon }{2}} \]
   for a given constant $\gamma  \in \bigl(0, \frac{1}{8\sqrt{2}}\bigr]$ (see \eqref{eq:gammadef}).
In this setting, $q_0, q_1 \in \Zc$, $|q_0| \leq Q_{\max}$,  $|q_1| \leq 2 Q_{\max}$, and $\sqrt{P^{1- \alpha}} \in \Zc^+$, based on our definitions $P \defeq \max_{k}\{ 2^{2 m_{kk}} \}$ and $\sqrt{P^{\alpha_{k\ell}}} =   2^{m_{k\ell}} ,  k, \ell =1,2$.

Let $\hat{q}_0$ and $ \hat{q}_1$ be the  estimates of  $q_0$ and  $ q_1$, respectively,  from the observation $y_1$ expressed in \eqref{eq:rewritteny231}.  Specifically, we use  an estimator that  minimizes
\[ |y_1 -  P^{{ \frac{ \alpha/2 + \epsilon }{2}}} \cdot  2\gamma \cdot  (  \sqrt{P^{ 1- \alpha }}  g_0 \hat{q}_0 + g_1 \hat{q}_1 ) |.   \]
We now consider the minimum distance
  \begin{align}
 d_{\min}  (g_0, g_1)    \defeq    \min_{\substack{ q_0, q_0' \in \Zc  \cap [- Q_{\max},    Q_{\max}]  \\  q_1, q_1' \in \Zc  \cap [- 2Q_{\max},   2 Q_{\max}]   \\  (q_0, q_1) \neq  (q_0', q_1')  }}  |  \sqrt{P^{ 1- \alpha }}  g_0  (q_0 - q_0') +  g_1 (q_1 - q_1')  | .    \label{eq:minidis111}
 \end{align}
 between the  signals generated by  $(q_0, q_1)$ and $(q_0', q_1') $.
 Later on, Lemma~\ref{lm:distance3432} (see below) shows that the minimum distance $d_{\min}$  is sufficiently large for almost all the channel  coefficients $\{h_{k\ell}\} \in (1, 2]^{2\times 3} $ when $P$ is large,  with the signal design in \eqref{eq:constellationGsym1}-\eqref{eq:constellationGsym1u}
and \eqref{eq:para34111}-\eqref{eq:xvkkk1341}. 
Let us now provide  Lemma~\ref{lm:distance3432}, the proof of which is based on the result of Lemma~\ref{lm:NMb2} (see Section~\ref{sec:example}).  
\begin{lemma}  \label{lm:distance3432}
Consider  the signal design in \eqref{eq:constellationGsym1}-\eqref{eq:constellationGsym1u} and \eqref{eq:para34111}-\eqref{eq:xvkkk1341}  in the  case  of $\alpha \in [2/3, 1]$. Let $\delta \in (0, 1]$ and  $\epsilon >0$.  Then the minimum distance $d_{\min}$ defined in \eqref{eq:minidis111} is bounded by
 \begin{align}
d_{\min}    \geq   \delta P^{- \frac{3\alpha/2 -1 }{2}}    \label{eq:distancegeq}
 \end{align}
for all  the channel  coefficients $\{h_{k\ell}\} \in (1, 2]^{2\times 3} \setminus \Ho$, where  $\Ho \subseteq (1,2]^{2\times 3}$ is an outage set and the Lebesgue measure of the outage set, denoted by $\mathcal{L}(\Ho)$,  satisfies  
 \begin{align}
\mathcal{L}(\Ho) \leq 1792 \delta   \cdot     P^{ - \frac{ \epsilon  }{2}} .    \label{eq:LebmeaB}
 \end{align}
 \end{lemma}
 \begin{proof}
 Consider the case of   $\alpha \in [2/3, 1]$. Let us set
\[ \beta \defeq    \delta P^{- \frac{3\alpha/2 -1 }{2}} , \quad A_0 \defeq   \sqrt{P^{ 1- \alpha }} ,  \quad  A_1 \defeq  1 \]
\[ Q_0 \defeq 2 Q_{\max},    \quad Q_1\defeq 4 Q_{\max},  \quad  Q_{\max} \defeq P^{ \frac{ \alpha/2 - \epsilon }{2}}       \]
for some $\epsilon >0$ and $\delta \in (0, 1]$.  Recall that $g_0\defeq  h_{11}    h_{23} h_{12}$,  $g_1\defeq   h_{12}h_{21} h_{13}$. Let $\tau \defeq 8$.
 Define the event
\begin{align}
  B(q_0, q_1)  \defeq \{ (g_0, g_1)  \in (1, \tau]^2 :  |  A_0 g_0 q_0 + A_1 g_1 q_1  | < \beta \}    \non
\end{align}
and set 
\begin{align}
B  \defeq   \bigcup_{\substack{ q_0, q_1 \in \Zc:  \\  (q_0, q_1) \neq  0,  \\  |q_k| \leq Q_k  \ \forall k }}  B(q_0, q_1) .    \non
\end{align}
We  now bound  the Lebesgue measure of $B $  by using the result of Lemma~\ref{lm:NMb2} (see Section~\ref{sec:example}), given as 
\begin{align}
\Lc (B )  & \leq   8 (\tau -1) \beta \min\{ \frac{ 8 {Q^2}_{\max}}{ 1},  \frac{ 8 {Q^2}_{\max}}{  \sqrt{P^{ 1- \alpha }}},   \frac{2 Q_{\max} \tau}{1},   \frac{4 Q_{\max} \tau}{ \sqrt{P^{ 1- \alpha }}}\} \non\\
& \leq   8 (\tau -1)\beta \cdot Q_{\max} \cdot \min\{8Q_{\max},  \frac{ 8Q_{\max}}{\sqrt{P^{ 1- \alpha }}},   2 \tau,   \frac{ 4\tau}{ \sqrt{P^{ 1- \alpha }}}\} \non\\ 
& \leq   8 (\tau -1)\beta \cdot Q_{\max} \cdot P^{\frac{ \alpha -1 }{2}}  \cdot \min\{8Q_{\max},   4\tau \} \non\\ 
& \leq   8 (\tau -1)\beta \cdot Q_{\max} \cdot P^{\frac{ \alpha -1 }{2}}  \cdot    4\tau  \non\\ 
& = 32 \tau (\tau -1) \beta \cdot  P^{\frac{3 \alpha/2 -1 -  \epsilon }{2}} \non \\
& = 32 \tau (\tau -1) \delta   \cdot     P^{ - \frac{ \epsilon  }{2}}  \non\\
& = 1792 \delta   \cdot     P^{ - \frac{ \epsilon  }{2}}.    \label{eq:LeBmeasure}
\end{align}
From our definition,  the set  $B$ is the collection of  $(g_0, g_1)$, where $(g_0, g_1)  \in (1,\tau]^2$.   For any $(g_0, g_1) \in B$,  there exists at least one pair  $(q_0, q_1) \in \{ q_0, q_1:  q_0, q_1 \in \Zc, (q_0, q_1) \neq  0,  |q_0| \leq Q_0 , |q_1| \leq Q_1 \}$ such that $| A_0  g_0q_0 + A_1 g_1 q_1  | <   \delta P^{- \frac{3\alpha/2 -1 }{2}}$. Thus, we can consider the set  $B$ as an  outage set. 
For any pair $(g_0, g_1)$  that is outside  the outage set $B$, we have the following conclusion: 
  \[d_{\min} (g_0, g_1)   \geq    \delta P^{- \frac{3\alpha/2 -1 }{2}} , \quad \text{for}  \quad  (g_0, g_1)\notin B .\]

Let us now define  $\Ho$  as a set of the  $(h_{11}, h_{21},  h_{12}, h_{22}, h_{13}, h_{23} ) \in (1, 2]^{2\times 3}$   such that the corresponding pairs $(g_0, g_1)$ are in the outage set $B$, i.e., 
\[ \Ho \defeq \{  (h_{11}, h_{21}, h_{12}, h_{22},  h_{13}, h_{23} ) \in (1, 2]^{2\times 3} :      (g_0, g_1) \in B  \} . \]
At this point, with the connection between  $\Ho$ and $B$, we can  bound the Lebesgue measure of $\Ho$  as
\begin{align}
 \Lc (\Ho ) &\!  = \!  \int_{h_{11}=1}^2      \!  \int_{h_{21}=1}^2    \!   \int_{h_{12}=1}^2   \!   \int_{h_{22}=1}^2    \!  \int_{h_{13}=1}^2    \!  \int_{h_{23}=1}^2  \! \mathbbm{1}_{\Ho}   \!  (h_{11}, \!   h_{21},\!   h_{12}, \!  h_{22},  \!  h_{13}, \!   h_{23} )    d h_{23}  d h_{13}    d h_{22}   d h_{12}   d h_{21}     d h_{11}           \label{eq:LeBmeasure882441}\\  
& =    \int_{h_{11}=1}^2   \int_{h_{21}=1}^2  \int_{h_{12}=1}^2 \int_{h_{22}=1}^2    \int_{h_{13}=1}^2  \int_{h_{23}=1}^2     \mathbbm{1}_{B}  (h_{11}    h_{23} h_{12}, h_{12}h_{21} h_{13}) d h_{23} d h_{13} d h_{22}  d h_{12} d h_{21}  d h_{11}           \non\\ 
& \leq   \int_{h_{11}=1}^2   \int_{h_{21}=1}^2  \int_{h_{12}=1}^2 \int_{h_{22}=1}^2      \int_{g_{1}=1}^{\tau}   \int_{g_{0}=1}^{\tau}    \mathbbm{1}_{B}  (g_0, g_{1})  \cdot  \frac{1}{h_{11}{h^2}_{12}h_{21}} d g_{0} d g_{1} d h_{22}  d h_{12} d h_{21}  d h_{11}                             \non\\ 
& \leq \int_{h_{11}=1}^2   \int_{h_{21}=1}^2  \int_{h_{12}=1}^2 \int_{h_{22}=1}^2     \int_{g_{1}=1}^{\tau}   \int_{g_{0}=1}^{\tau}    \mathbbm{1}_{B}  (g_0, g_{1})  d g_{0} d g_{1} d h_{22}  d h_{12} d h_{21}  d h_{11}             \non \\
& =     \int_{h_{11}=1}^2   \int_{h_{21}=1}^2  \int_{h_{12}=1}^2 \int_{h_{22}=1}^2      \mathcal{L}(B)  d h_{22}  d h_{12} d h_{21}  d h_{11}                   \non\\ 
& \leq    \int_{h_{11}=1}^2   \int_{h_{21}=1}^2  \int_{h_{12}=1}^2 \int_{h_{22}=1}^2      1792\delta   \cdot     P^{ - \frac{ \epsilon  }{2}}    d h_{22}  d h_{12} d h_{21}  d h_{11}           \label{eq:LeBmeasure0993}  \\     
& =    1792\delta   \cdot     P^{ - \frac{ \epsilon  }{2}}               \label{eq:LeBmeasure9999}  
\end{align}
where  \eqref{eq:LeBmeasure0993} follows from the result in \eqref{eq:LeBmeasure};
in the above derivations we use the following definitions
\begin{subnumcases}
{ \mathbbm{1}_{\Ho}  (h_{11}, h_{21}, h_{12}, h_{22},  h_{13}, h_{23}) 
 =} 
    1  &     if   \ $(h_{11}, h_{21}, h_{12}, h_{22},  h_{13}, h_{23} ) \in  \Ho$            			\non \\
     0  &  if    \  $(h_{11}, h_{21}, h_{12}, h_{22},  h_{13}, h_{23} ) \notin  \Ho$ \non
\end{subnumcases}
and 
\begin{subnumcases}
{ \mathbbm{1}_{B}  (g_0, g_1) 
 =} 
    1  &     if   \ $(g_0, g_1) \in  B$            			\non \\
     0  &  if    \  $(g_0, g_1) \notin  B $ . \non
\end{subnumcases} 
Then, we complete  the proof of  Lemma~\ref{lm:distance3432}.
\end{proof}

 The result of Lemma~\ref{lm:distance3432} reveals  that,  the minimum distance $d_{\min}$  is sufficiently large
for all  the channel  coefficients $\{h_{k\ell}\} \in (1, 2]^{2\times 3} $ except for a bounded set   $\Ho \subseteq (1,2]^{2\times 3}$ whose Lebesgue measure  satisfies  \[ \mathcal{L}(\Ho)  \to 0, \quad \text {as}\quad  P\to \infty .\] 
In the following, we will consider the channel  coefficients $\{h_{k\ell}\} \in (1, 2]^{2\times 3} $ that are not in the outage set   $\Ho$. With this channel condition, the minimum distance $d_{\min}$   is bounded by  \[d_{\min}    \geq  \delta P^{- \frac{3\alpha/2 -1 }{2}}. \]  

Let us now  go back to the expression of $y_1$ in  \eqref{eq:rewritteny231}, which can also be described as
\begin{align}
 y_1    =  &  P^{{ \frac{ \alpha/2 + \epsilon }{2}}}  \cdot 2\gamma  \underbrace{(\sqrt{P^{ 1- \alpha }}  g_0 q_0 + g_1 q_1 )}_{ \defeq x_{s}  }      \non\\
&  +   \sqrt{P^{ 1 - \alpha}}  \bigl(\underbrace{ h_{11}    h_{23} h_{12}  v_{1,p}   +    \frac{1}{\sqrt{P^{ 1 - \alpha}}}h_{12}h_{21} h_{13} v_{2,p}  \bigr)}_{\defeq  \tilde{g} }   +  z_{1}   \non\\
 =   & P^{{ \frac{ \alpha/2 + \epsilon }{2}}}  \cdot 2\gamma \cdot x_{s}   +  \sqrt{P^{ 1 - \alpha}} \tilde{g}  +  z_{1}     \label{eq:yk182375}  
\end{align}
where 
\[x_{s}  \defeq   \sqrt{P^{ 1- \alpha }}  g_0 q_0 + g_1 q_1 \] 
and $\tilde{g} \defeq h_{11}    h_{23} h_{12}  v_{1,p}   +    \frac{1}{\sqrt{P^{ 1 - \alpha}}}h_{12}h_{21} h_{13} v_{2,p}   $, with  
\[ |\tilde{g} | \leq  \tilde{g}_{\max} \defeq \sqrt{2}  \quad  \forall   \tilde{g} \]
In the first step we will  decode the sum $x_{s}  =   \sqrt{P^{ 1- \alpha }}  g_0 q_0 + g_1 q_1 $ 
from $y_1$ by treating other signals as noise. This step is called as noise removal.  After correctly decoding  $x_{s} $, we can recover $q_0$  and $q_1$ from $x_{s} $  because $g_0$ and $g_1$ are rationally independent. This step is called as signal separation (cf.~\cite{MGMK:14}). 
Given that  the  minimum distance for $x_{s}$ is  $d_{\min}$ defined in \eqref{eq:minidis111},  Lemma~\ref{lm:distance3432} reveals that this minimum distance is bounded by  $d_{\min}    \geq   \delta P^{- \frac{3\alpha/2 -1 }{2}} $ when the channels are not in the outage set, that is,  $\{h_{k\ell}\} \notin \Ho$. Define $\hat{x}_s$ as the estimate for $ x_s$ from $y_1$ by choosing the point close to $ x_s$.
Then,  the error  probability for decoding $x_{s}$ from $y_1$  in \eqref{eq:yk182375} is 
 \begin{align}
    \text{Pr} [ x_s \neq \hat{x}_s ]  
   \leq  &  \text{Pr} \Bigl[   | z_1  + P^{ \frac{1 - \alpha}{2}} \tilde{g}|  >    P^{{ \frac{ \alpha/2 + \epsilon }{2}}}  \cdot 2\gamma   \cdot \frac{d_{\text{min}} }{2}  \Bigr]   \non \\
   \leq  &     \text{Pr} [  z_1   >  P^{{ \frac{ \alpha/2 + \epsilon }{2}}}  \cdot 2\gamma   \cdot \frac{d_{\text{min}} }{2}   -  P^{ \frac{1 - \alpha}{2}} \tilde{g}_{\max}   ]   \non\\&+  \text{Pr} [  z_1   >  P^{{ \frac{ \alpha/2 + \epsilon }{2}}}  \cdot 2\gamma   \cdot \frac{d_{\text{min}} }{2}   -  P^{ \frac{1 - \alpha}{2}} \tilde{g}_{\max}   ]  \label{eq:error2256cc} \\
    =    & 2  \cdot     {\bf{Q}} \bigl(  P^{{ \frac{ \alpha/2 + \epsilon }{2}}}  \cdot 2\gamma   \cdot \frac{d_{\text{min}} }{2}   -  P^{ \frac{1 - \alpha}{2}} \tilde{g}_{\max}  \bigr)    \non   \\ 
     \leq  &  2  \cdot     {\bf{Q}} \bigl(  P^{ \frac{1 - \alpha}{2}} ( \gamma \delta  P^{ \frac{ \epsilon}{2}}   -  \sqrt{2})  \bigr)     \label{eq:error9982cc} 
  \end{align}
  where the ${\bf{Q}}$-function is defined as ${\bf{Q}}(a )  \defeq  \frac{1}{\sqrt{2\pi}} \int_{a}^{\infty}  \exp( -\frac{ s^2}{2} ) d s$;
\eqref{eq:error2256cc} follows from that   $ |\tilde{g} | \leq  \tilde{g}_{\max} \defeq \sqrt{2},    \forall   \tilde{g}$;
\eqref{eq:error9982cc} is from  that $d_{\min}    \geq    \delta P^{- \frac{3\alpha/2 -1 }{2}}  $.
By using the identity that  $ {\bf{Q}} (a ) \leq   \frac{1}{2}\exp ( -  a^2 /2 ),  \    \forall a \geq 0$,  and together with \eqref{eq:error9982cc}, we can  have the following bound  
  \begin{align}
  \text{Pr} [ x_s \neq \hat{x}_s ]   &\leq    \exp \Bigl(    -  \frac{ P^{1 - \alpha}  \bigl(  \gamma \delta  P^{ \frac{ \epsilon}{2}}   -  \sqrt{2} \bigr)^2 }{2} \Bigr)  \non
   \end{align}
 when  $   \gamma \delta  P^{ \frac{ \epsilon}{2}}   -  \sqrt{2} \geq 0 $.   
Therefore, when $P\to \infty$, it implies that  $   \gamma \delta  P^{ \frac{ \epsilon}{2}}   -  \sqrt{2} \geq 0 $, which then gives the following conclusion on the  error probability for decoding $x_{s}$ from $y_1$ 
 \begin{align}
 \text{Pr} [ x_s \neq \hat{x}_s ] \to 0  \quad  \text{as} \quad   P\to \infty.    \label{eq:error885256}                           
 \end{align}
After decoding   $x_{s}  =   \sqrt{P^{ 1- \alpha }}  g_0 q_0 + g_1 q_1 $ correctly from $y_1$,  we can recover the symbols $  q_0   $ and  $q_1  $  because  $g_0$ and $g_1$ are rationally independent.

In the second step,  the decoded $x_{s}$ can be removed from $y_1$  and then  $v_{1,p}$ can be decoded from 
\begin{align}
  y_1  - P^{{ \frac{ \alpha/2 + \epsilon }{2}}}  \cdot 2\gamma \cdot x_{s} =   \sqrt{P^{ 1 - \alpha}} h_{11}    h_{23} h_{12}  v_{1,p}  +    h_{12}h_{21} h_{13} v_{2,p}   +  z_{1} .    \non
\end{align}
Given that   $v_{1,p}, v_{2,p} \in    \Omega (\xi   =\gamma \cdot \frac{ 1}{Q},   \   Q = P^{ \frac{ 1 - \alpha - \epsilon}{2}} )$  and $h_{12}h_{21} h_{13} v_{2,p} \leq  \frac{1}{\sqrt{2}}$,  the result of Lemma~\ref{lm:AWGNic}  reveals that  the error probability for decoding $v_{1,p}$  is 
 \begin{align}
  \text{Pr} [ v_{1,p} \neq \hat{v}_{1,p} ]     \to 0  \quad  \text{as} \quad   P\to \infty . \label{eq:error155525}                           
 \end{align}
 
At this point, it holds true that  the error probability of estimating  $\{v_{1,c}, v_{1,p} \}$ from $y_1$  is
 \begin{align}
 \text{Pr} [  \{ v_{1,c} \neq \hat{v}_{1,c} \} \cup  \{ v_{1,p} \neq \hat{v}_{1,p} \}  ]  \to 0         \quad \text {as}\quad  P\to \infty    \label{eq:error1c1p0595}
 \end{align}
 for  all the channel coefficients $\{h_{k\ell}\} \in (1, 2]^{2\times 3} \setminus \Ho$,  with Lebesgue measure  $\mathcal{L}(\Ho)$ satisfying  \[ \mathcal{L}(\Ho)  \to 0, \quad \text {as}\quad  P\to \infty .\] 
Due to the symmetry, the case of $k=2$ can be proved in a similar way.

\section{Converse }   \label{sec:converse}
In this section, we provide the converse proof of the secure sum GDoF in Theorem~\ref{thm:IChGDoF}, focusing on the two-user symmetric  Gaussian interference channel with a helper defined in Section~\ref{sec:system}. 
The converse proof is based on the following two steps: 

1)  The secure capacity region of  the interference channel with a helper and with  secrecy constraints, denoted by $C$,  is outer bounded by the capacity region of the interference channel with a  helper but without  secrecy constraints (denoted by $C_{h, ns}$),  i.e., 
\[  C \subseteq  C_{h, ns}\]
due to the fact that  secrecy constraints will not enlarge the capacity region.

 2) The capacity region of the interference channel with a helper is outer bounded by the capacity region of the interference channel without a helper (denoted by $C_{nh, ns}$), i.e., 
\[  C_{h, ns} \subseteq   C_{nh, ns}\]
because the helper's  signal that is independent of the other transmitters' signals, will not enlarge the capacity region for this setting without secrecy constraints.

Therefore, it implies that 
\[C \subseteq C_{nh, ns}. \]  
In other words,  the capacity region $C_{nh, ns}$ (and the GDoF region respectively) of the interference channel without a helper and without secrecy constraints (see \cite{ETW:08}),    will sever as the outer bound of the secure capacity region $C$ (and the secure GDoF region respectively) of the interference channel with a helper and with secrecy constraints.

\section{Conclusion}   \label{sec:conclusion}

For the two-user symmetric Gaussian interference channel, this work revealed an interesting observation that   adding a helper can  \emph{totally} remove the secrecy constraints, in terms of GDoF performance.
In the proposed scheme,  the cooperative jamming and a careful signal design are used such that  the jamming signal of  the helper  is aligned at a specific direction and power level with the information signals of the transmitters. It turns out that,  the penalty in GDoF due to the secrecy constraints can be totally removed.  In the rate analysis, the estimation approaches of noise removal and signal separation due to the rational independence  are used. 
In the future work, we will extend the result to the other communication channels with a helper.

\appendices

\section{Proof of Lemma~\ref{lm:rateerror322}  \label{sec:rateerror322} }

In this section we will provide the proof of Lemma~\ref{lm:rateerror322}. 
We first focus on the proof for the first user ($k=1$). 
In the following, $v_{2,c} + u_3$ and  $v_{1,c}$ can be estimated from $y_1$ by using a successive decoding method with vanishing error probability,  where $y_1$ is expressed in \eqref{eq:yvk122}.
At first,   $v_{2,c} + u_3 \in   2\cdot \Omega ( \xi =  \gamma_{v_{2,c}} \cdot \frac{ 1}{Q} ,   \   Q =  P^{ \frac{1  - \epsilon }{2}} )$  can be estimated from $y_1$  by treating the other signals as noise, where  $y_1$ in \eqref{eq:yvk122} can be described as 
\begin{align}
 y_1  =    \sqrt{P^{ \alpha }} h_{12}h_{21} h_{13}  (   v_{2,c}  +  u_3)     +  \sqrt{P} h_{11}    h_{23} h_{12} v_{1,c}   +  z_{1}  .  \label{eq:yvk101322}
\end{align}
In this  case with $ \alpha \geq  2$, we have $ | h_{11}    h_{23} h_{12} v_{1,c} | \leq   \sqrt{2}$ for any realizations of $v_{1,c}$.
Let  $\hat{s}_{vu}$ be the estimate of $s_{vu} \defeq v_{2,c} + u_3$. From Lemma~\ref{lm:AWGNic} it is true that   
 \[\text{Pr} [s_{vu}  \neq  \hat{s}_{vu}] \to 0, \quad \text {as}\quad  P\to \infty.\]

Once   $v_{2,c} + u_3$ is decoded, it can be removed from $y_1$. After that,  $v_{1,c} \in   \Omega ( \xi =  \gamma_{v_{1,c}} \cdot \frac{ 1}{Q} ,   \   Q =  P^{ \frac{1  - \epsilon }{2}} )$  can be estimated   from  the following observation  \[y_1-    \sqrt{P^{ \alpha }} h_{12}h_{21} h_{13}  (   v_{2,c}  +  u_3) =   \sqrt{P} h_{11}    h_{23} h_{12} v_{1,c}   +  z_{1}    \]  with vanishing  error probability, i.e.,   
  \begin{align}
 \text{Pr} [ v_{1,c} \neq  \hat{v}_{1,c}] &   \to 0    \quad \text {as }\quad  P \to \infty   \label{eq:error2950322}
\end{align}
which completes the proof for the case of $k=1$. Similarly,  the case of $k=2$ can be proved due to the symmetry.

\section{Proof of Lemma~\ref{lm:rateerror132}  \label{sec:rateerror132} }

In this section we will provide the proof of Lemma~\ref{lm:rateerror132}, which is similar to that of Lemma~\ref{lm:rateerror341}.
In this case with $\alpha \in [1, 2]$, the two symbols 
\begin{align}
v_{1,c} &\in \Omega ( \xi =  2\gamma \cdot \frac{ 1}{Q} ,   Q =  P^{ \frac{ \alpha/2 - \epsilon }{2}} ) \non \\
 v_{2,c}  +  u_3 &\in 2 \cdot \Omega ( \xi =  2\gamma \cdot \frac{ 1}{Q},   \   Q =  P^{ \frac{ \alpha/2 - \epsilon }{2}} ) \non
\end{align}
will be estimated from $y_1$   by  using the approaches of noise removal and signal separation.  The expression of $y_1$ is given in \eqref{eq:yvk122}, which can be rewritten as 
\begin{align}
y_{1} &=    \sqrt{P} h_{11}    h_{23} h_{12} v_{1,c}   +    \sqrt{P^{ \alpha }} h_{12}h_{21} h_{13}  (   v_{2,c}  +  u_3)    +  z_{1}  \non   \\
         &=    \sqrt{P^{ 1-  \alpha/2 +  \epsilon}} \cdot 2\gamma \cdot ( \bar{g}_0 \bar{q}_0 + \sqrt{P^{ \alpha -1 }}\bar{g}_1 \bar{q}_1 )   +  z_{1}    \non\\
         &=    \sqrt{P^{ 1-  \alpha/2 +  \epsilon}}  \cdot 2\gamma \cdot \bar{x}_s   +  z_{1}    \label{eq:yvk4835132}
\end{align}
where  $\gamma  \in \bigl(0, \frac{1}{8\sqrt{2}}\bigr]$ ,  $\bar{x}_s  \defeq ( \bar{g}_0 \bar{q}_0 + \sqrt{P^{ \alpha -1 }}\bar{g}_1 \bar{q}_1 )$, $\bar{g}_0\defeq  h_{11}    h_{23} h_{12} $,  $\bar{g}_1\defeq    h_{12}h_{21} h_{13} $   and 
    \[ \bar{q}_0  \defeq  \frac{Q_{\max}}{2\gamma} \cdot   v_{1,c}  ,    \quad  \bar{q}_1  \defeq  \frac{Q_{\max}}{2\gamma} \cdot   (   v_{2,c}  +  u_3),    \quad  Q_{\max} \defeq P^{ \frac{ \alpha/2 - \epsilon }{2}} . \]
In this setting,  $\bar{q}_0, \bar{q}_1 \in \Zc$,   $|\bar{q}_0| \leq Q_{\max}$, $|\bar{q}_1| \leq 2 Q_{\max}$,   $ \sqrt{P^{ \alpha -1 }} \in  \Zc^+$ based on  our definitions.
Let us define  the minimum distance for $\bar{x}_s$ as
  \begin{align}
 \bar{d}_{\min}  (\bar{g}_0, \bar{g}_1)   \defeq    \min_{\substack{ \bar{q}_0, \bar{q}_0', \in \Zc  \cap [- Q_{\max},    Q_{\max}]  \\  \bar{q}_1, \bar{q}_1' \in \Zc  \cap [- 2Q_{\max},   2 Q_{\max}]   \\  (\bar{q}_0, \bar{q}_1) \neq  (\bar{q}_0', \bar{q}_1'  )}}  | \bar{g}_0  (\bar{q}_0 - \bar{q}_0') + \sqrt{P^{ \alpha - 1}} \bar{g}_1 (\bar{q}_1 - \bar{q}_1')  |.   \label{eq:minidis111132}
 \end{align}
A lower bound on the minimum distance $\bar{d}_{\min}$ is given in the following  lemma.

\begin{lemma}  \label{lm:distance132}
 Consider the signal design in \eqref{eq:constellationGsym1}-\eqref{eq:constellationGsym1u} and \eqref{eq:para13211}-\eqref{eq:xvkkk1132} for the case with $\alpha \in [1, 2]$. Let $\delta \in (0, 1]$ and  $\epsilon >0$.    The minimum distance $\bar{d}_{\min}$ defined in \eqref{eq:minidis111132} is bounded by
 \begin{align}
d_{\min}    \geq   \delta P^{- \frac{1- \alpha/2 }{2}}   \label{eq:distancegeq132}
 \end{align}
for all  the channel  coefficients $\{h_{k\ell}\} \in (1, 2]^{2\times 3} \setminus \Hob$, where $\Hob \subseteq (1,2]^{2\times 3}$ is an outage set  and the  the Lebesgue measure of the this outage set, denoted by $\mathcal{L}(\Hob)$,  satisfies  
 \begin{align}
\mathcal{L}(\Hob) \leq  1792\delta   \cdot     P^{ - \frac{ \epsilon  }{2}}  .   \label{eq:LebmeaB132}
 \end{align}
 \end{lemma}
 \begin{proof}
The proof of this lemma is  similar to that of the Lemma~\ref{lm:distance3432}.
Considering the case  of  $\alpha \in [1, 2]$,  let  us set
\[ \bar{\beta} \defeq   \delta P^{- \frac{1- \alpha/2 }{2}} , \quad \bar{A}_0 \defeq  1,  \quad  \bar{A}_1 \defeq    \sqrt{P^{  \alpha -1 }} \]
\[ Q_0 \defeq 2 Q_{\max},    \quad Q_1\defeq 4 Q_{\max},  \quad  Q_{\max} \defeq P^{ \frac{ \alpha/2 - \epsilon }{2}}       \]
for some $\epsilon >0$ and $\delta \in (0, 1]$. Recall that $ \bar{g}_0\defeq   h_{11}    h_{23} h_{12}$ and $ \bar{g}_1\defeq    h_{12}h_{21} h_{13}$. Let  $\tau \defeq 8$. 
We define the event
\begin{align}
\bar{B}(\bar{q}_0, \bar{q}_1)  \defeq \{ (\bar{g}_0, \bar{g}_1)  \in (1, \tau]^2 :  |  \bar{A}_0 \bar{g}_0\bar{q}_0 + \bar{A}_1 \bar{g}_1 \bar{q}_1 | < \bar{ \beta} \}       \label{eq:Boutage11132}
\end{align}
and set 
\begin{align}
\bar{B}  \defeq   \bigcup_{\substack{ \bar{q}_0, \bar{q}_1 \in \Zc:  \\  (\bar{q}_0, \bar{q}_1) \neq  0,  \\  |\bar{q}_k| \leq Q_k  \ \forall k }}  B(\bar{q}_0, \bar{q}_1) .   \label{eq:Boutage22132}
\end{align}
By using Lemma~\ref{lm:NMb2},  we bound the Lebesgue measure of $\bar{B}$ as 
\begin{align}
\Lc (\bar{B})  & \leq   8 (\tau -1) \beta \min\{ \frac{ 8 {Q^2}_{\max}}{ \sqrt{P^{  \alpha -1 }}},  \frac{ 8 {Q^2}_{\max}}{  1},   \frac{2 Q_{\max} \tau}{ \sqrt{P^{  \alpha -1 }}},   \frac{4 Q_{\max} \tau}{1}\} \non\\
& \leq   8 (\tau -1)\beta \cdot Q_{\max} \cdot \min\{ \frac{ 8Q_{\max}}{ \sqrt{P^{  \alpha -1 }}}, 8Q_{\max},     \frac{ 2\tau}{  \sqrt{P^{  \alpha -1 }}}, 4 \tau  \} \non\\ 
& \leq   8 (\tau -1)\beta \cdot Q_{\max} \cdot P^{\frac{ 1- \alpha  }{2}}  \cdot \min\{8Q_{\max},   4\tau \} \non\\ 
& \leq   8 (\tau -1)\beta \cdot Q_{\max} \cdot P^{\frac{ 1- \alpha  }{2}}  \cdot    4\tau  \non\\ 
& = 32 \tau (\tau -1) \beta \cdot  P^{\frac{1-  \alpha/2  -  \epsilon }{2}} \non \\
& = 32 \tau (\tau -1) \delta   \cdot     P^{ - \frac{ \epsilon  }{2}}   \non\\
& = 1792 \delta   \cdot     P^{ - \frac{ \epsilon  }{2}}.   \label{eq:LeBmeasure132}
\end{align}
We can consider  set  $\bar{B}$ as an outage set. 
For any  pair $(\bar{g}_0, \bar{g}_1)$  that is outside the outage set $\bar{B}$, we can conclude that  $\bar{d}_{\min} (\bar{g}_0, \bar{g}_1)   \geq   \delta P^{- \frac{1- \alpha/2 }{2}}$. 
Let us now define  $\Hob$  as a set  of the  $(h_{11}, h_{21}, h_{12}, h_{22},  h_{13}, h_{23} ) \in (1, 2]^{2\times 3}$   such that the corresponding pairs $(\bar{g}_0, \bar{g}_1)$ are in the outage set $B$, that is,
\[ \Hob \defeq \{  (h_{11}, h_{21}, h_{12}, h_{22},  h_{13}, h_{23} ) \in (1, 2]^{2\times 3} \!:      (\bar{g}_0, \bar{g}_1) \in \bar{B}  \} . \]
With the connection between $\Hob$ and $\bar{B}$,  one can follow the steps in \eqref{eq:LeBmeasure882441}-\eqref{eq:LeBmeasure9999} to bound  the Lebesgue measure of $\Hob$  as: 
\begin{align}
\Lc (\Hob )  \leq  1792 \delta   \cdot     P^{ - \frac{ \epsilon  }{2}} .              \label{eq:LeBmeasure9999132}  
\end{align}
 \end{proof}

Let us now go back to the expression of $y_1$ in \eqref{eq:yvk4835132}, i.e., \[y_{1}   =     \sqrt{P^{ 1-  \alpha/2 +  \epsilon}}  \cdot 2\gamma \cdot \bar{x}_s   +  z_{1}   .\]
Note that the minimum distance for $\bar{x}_{s}$ is  $\bar{d}_{\min}$ defined in \eqref{eq:minidis111132}.  Lemma~\ref{lm:distance132} shows that this minimum distance for $\bar{x}_{s}$ is bounded by  $\bar{d}_{\min}    \geq   \delta P^{- \frac{1- \alpha/2 }{2}} $ when the channels are not in the outage set, i.e.,  $\{h_{k\ell}\} \notin \Hob$.  This  implies that  we can estimate $\bar{x}_s$ from $y_{1}$ with vanishing  error  probability as  $P\to \infty$. 
Since  $\bar{g}_0$ and $ \bar{g}_1$ are rationally independent,  we can recover  the symbols $ \bar{q}_0 $  and    $\bar{q}_1 $ from $\bar{x}_s  = ( \bar{g}_0 \bar{q}_0 + \sqrt{P^{ \alpha -1 }}\bar{g}_1 \bar{q}_1 )$ after  decoding  $\bar{x}_{s} $.
Then, the error probability for decoding $v_{1,c}$  can be concluded as  
 \begin{align}
  \text{Pr} [ v_{1,c} \neq \hat{v}_{1,c} ]     \to 0  \quad  \text{as} \quad   P\to \infty  \label{eq:error155525132}                           
 \end{align}
which  is true  for all the channel coefficients $\{h_{k\ell}\} \in (1, 2]^{2\times 3} \setminus \Ho$,   with Lebesgue measure  $\mathcal{L}(\Hob)$ satisfying  \[ \mathcal{L}(\Hob)  \to 0, \quad \text {as}\quad  P\to \infty. \] 
Similarly, the result in \eqref{eq:error155525132} can also be extended to the case of $k=2$   due to the symmetry.


\begin{thebibliography}{10}
\providecommand{\url}[1]{#1}
\csname url@samestyle\endcsname
\providecommand{\newblock}{\relax}
\providecommand{\bibinfo}[2]{#2}
\providecommand{\BIBentrySTDinterwordspacing}{\spaceskip=0pt\relax}
\providecommand{\BIBentryALTinterwordstretchfactor}{4}
\providecommand{\BIBentryALTinterwordspacing}{\spaceskip=\fontdimen2\font plus
\BIBentryALTinterwordstretchfactor\fontdimen3\font minus
  \fontdimen4\font\relax}
\providecommand{\BIBforeignlanguage}[2]{{%
\expandafter\ifx\csname l@#1\endcsname\relax
\typeout{** WARNING: IEEEtran.bst: No hyphenation pattern has been}%
\typeout{** loaded for the language `#1'. Using the pattern for}%
\typeout{** the default language instead.}%
\else
\language=\csname l@#1\endcsname
\fi
#2}}
\providecommand{\BIBdecl}{\relax}
\BIBdecl

\bibitem{Shannon:49}
C.~E. Shannon, ``Communication theory of secrecy systems,'' \emph{Bell System
  Technical Journal}, vol.~28, no.~4, pp. 656 -- 715, Oct. 1949.

\bibitem{Wyner:75}
A.~D. Wyner, ``The wire-tap channel,'' \emph{Bell System Technical Journal},
  vol.~54, no.~8, pp. 1355 -- 1378, Jan. 1975.

\bibitem{CsiszarKorner:78}
I.~Csisz\`ar and J.~K{\"o}rner, ``Broadcast channels with confidential
  messages,'' \emph{IEEE Trans. Inf. Theory}, vol.~24, no.~3, pp. 339 -- 348,
  May 1978.

\bibitem{CH:78}
S.~K. {Leung-Yan-Cheong} and M.~E. Hellman, ``Gaussian wiretap channel,''
  \emph{IEEE Trans. Inf. Theory}, vol.~24, no.~4, pp. 451 -- 456, Jul. 1978.

\bibitem{LMSY:08}
R.~Liu, I.~Maric, P.~Spasojevic, and R.~D. Yates, ``Discrete memoryless
  interference and broadcast channel with confidential messages: secrecy rate
  regions,'' \emph{IEEE Trans. Inf. Theory}, vol.~54, no.~6, pp. 2493 -- 2507,
  Jun. 2008.

\bibitem{KTW:08}
A.~Khisti, A.~Tchamkerten, and G.~W. Wornell, ``Secure broadcasting over fading
  channels,'' \emph{IEEE Trans. Inf. Theory}, vol.~54, no.~6, pp. 2453 -- 2469,
  Jun. 2008.

\bibitem{XCC:09}
J.~Xu, Y.~Cao, and B.~Chen, ``Capacity bounds for broadcast channels with
  confidential messages,'' \emph{IEEE Trans. Inf. Theory}, vol.~55, no.~10, pp.
  4529 -- 4542, Oct. 2009.

\bibitem{LLPS:10}
R.~Liu, T.~Liu, H.~V. Poor, and S.~Shamai, ``Multiple-input multiple-output
  {Gaussian} broadcast channels with confidential messages,'' \emph{IEEE Trans.
  Inf. Theory}, vol.~56, no.~9, pp. 4215 -- 4227, Sep. 2010.

\bibitem{TY:08cj}
E.~Tekin and A.~Yener, ``The general {Gaussian} multiple-access and two-way
  wiretap channels: {Achievable} rates and cooperative jamming,'' \emph{IEEE
  Trans. Inf. Theory}, vol.~54, no.~6, pp. 2735 -- 2751, Jun. 2008.

\bibitem{LP:08}
Y.~Liang and H.~V. Poor, ``Multiple access channels with confidential
  messages,'' \emph{IEEE Trans. Inf. Theory}, vol.~54, no.~3, pp. 976 -- 1002,
  Mar. 2008.

\bibitem{TekinYener:08d}
E.~Tekin and A.~Yener, ``The {Gaussian} multiple access wire-tap channel,''
  \emph{IEEE Trans. Inf. Theory}, vol.~54, no.~12, pp. 5747 -- 5755, Dec. 2008.

\bibitem{HY:09}
X.~He and A.~Yener, ``A new outer bound for the {Gaussian} interference channel
  with confidential messages,'' in \emph{Proc. 43rd Annu. Conf. Inf. Sci.
  Syst.}, Mar. 2009.

\bibitem{KG:15}
S.~Karmakar and A.~Ghosh, ``Approximate secrecy capacity region of an
  asymmetric {MAC} wiretap channel within 1/2 bits,'' in \emph{IEEE 14th
  Canadian Workshop on Information Theory}, Jul. 2015.

\bibitem{HKY:13}
X.~He, A.~Khisti, and A.~Yener, ``{MIMO} multiple access channel with an
  arbitrarily varying eavesdropper: {Secrecy} degrees of freedom,'' \emph{IEEE
  Trans. Inf. Theory}, vol.~59, no.~8, pp. 4733 -- 4745, Aug. 2013.

\bibitem{LBPSV:09}
Y.~Liang, A.~{Somekh-Baruch}, H.~V. Poor, S.~Shamai, and S.~Verdu, ``Capacity
  of cognitive interference channels with and without secrecy,'' \emph{IEEE
  Trans. Inf. Theory}, vol.~55, no.~2, pp. 604 -- 619, Feb. 2009.

\bibitem{PDT:09}
E.~Perron, S.~Diggavi, and E.~Telatar, ``On noise insertion strategies for
  wireless network secrecy,'' in \emph{Proc. Inf. Theory and App. Workshop
  {(ITA)}}, Feb. 2009.

\bibitem{LYT:08}
Z.~Li, R.~D. Yates, and W.~Trappe, ``Secrecy capacity region of a class of
  one-sided interference channel,'' in \emph{Proc. {IEEE} Int. Symp. Inf.
  Theory {(ISIT)}}, Jul. 2008.

\bibitem{YTL:08}
R.~D. Yates, D.~Tse, and Z.~Li, ``Secret communication on interference
  channels,'' in \emph{Proc. {IEEE} Int. Symp. Inf. Theory {(ISIT)}}, Jul.
  2008.

\bibitem{KGLP:11}
O.~O. Koyluoglu, H.~{El Gamal}, L.~Lai, and H.~V. Poor, ``Interference
  alignment for secrecy,'' \emph{IEEE Trans. Inf. Theory}, vol.~57, no.~6, pp.
  3323 -- 3332, Jun. 2011.

\bibitem{LLP:11}
R.~Liu, Y.~Liang, and H.~V. Poor, ``Fading cognitive multiple-access channels
  with confidential messages,'' \emph{IEEE Trans. Inf. Theory}, vol.~57, no.~8,
  pp. 4992 -- 5005, Aug. 2011.

\bibitem{XU:14}
J.~Xie and S.~Ulukus, ``Secure degrees of freedom of one-hop wireless
  networks,'' \emph{IEEE Trans. Inf. Theory}, vol.~60, no.~6, pp. 3359 -- 3378,
  Jun. 2014.

\bibitem{MDHS:14}
A.~Mukherjee, S.~A. Fakoorian, and A.~L.~S. J.~Huang, ``Principles of physical
  layer security in multiuser wireless networks: {A} survey,'' \emph{IEEE
  Communications Surveys \& Tutorials}, vol.~16, no.~3, pp. 1550 -- 1573, Aug.
  2014.

\bibitem{MM:14o}
P.~Mohapatra and C.~R. Murthy, ``Outer bounds on the secrecy rate of the 2-user
  symmetric deterministic interference channel with transmitter cooperation,''
  in \emph{Proc. 20th National Conference on Communications (NCC)}, 2014.

\bibitem{XU:15}
J.~Xie and S.~Ulukus, ``Secure degrees of freedom of {$K$}-user {Gaussian}
  interference channels: {A} unified view,'' \emph{IEEE Trans. Inf. Theory},
  vol.~61, no.~5, pp. 2647 -- 2661, May 2015.

\bibitem{MM:16}
P.~Mohapatra and C.~R. Murthy, ``On the capacity of the two-user symmetric
  interference channel with transmitter cooperation and secrecy constraints,''
  \emph{IEEE Trans. Inf. Theory}, vol.~62, no.~10, pp. 5664 -- 5689, Oct. 2016.

\bibitem{GTJ:15}
C.~Geng, R.~Tandon, and S.~A. Jafar, ``On the symmetric 2-user deterministic
  interference channel with confidential messages,'' in \emph{Proc. {IEEE}
  Global Conf. Communications {(GLOBECOM)}}, Dec. 2015.

\bibitem{GJ:15}
C.~Geng and S.~A. Jafar, ``Secure {GDoF} of {$K$}-user {Gaussian} interference
  channels: {When} secrecy incurs no penalty,'' \emph{IEEE Communications
  Letters}, vol.~19, no.~8, pp. 1287 -- 1290, Aug. 2015.

\bibitem{allerton:16}
J.~Chen, ``New results on the secure capacity of symmetric two-user
  interference channels,'' in \emph{Proc. Allerton Conf. Communication, Control
  and Computing}, Sep. 2016.

\bibitem{MXU:17}
P.~Mukherjee, J.~Xie, and S.~Ulukus, ``Secure degrees of freedom of one-hop
  wireless networks with no eavesdropper {CSIT},'' \emph{IEEE Trans. Inf.
  Theory}, vol.~63, no.~3, pp. 1898 -- 1922, Mar. 2017.

\bibitem{ChenAllerton:18}
J.~Chen, ``Secure communication over interference channel: To jam or not to
  jam?'' in \emph{Proc. Allerton Conf. Communication, Control and Computing},
  Oct. 2018.

\bibitem{ChenIC:18}
------, ``Secure communication over interference channel: {To} jam or not to
  jam?'' Oct. 2018, available on ArXiv: https://arxiv.org/pdf/1810.13256.pdf.

\bibitem{ETW:08}
R.~H. Etkin, D.~N.~C. Tse, and H.~Wang, ``Gaussian interference channel
  capacity to within one bit,'' \emph{IEEE Trans. Inf. Theory}, vol.~54,
  no.~12, pp. 5534 -- 5562, Dec. 2008.

\bibitem{FW:16}
R.~Fritschek and G.~Wunder, ``Towards a constant-gap sum-capacity result for
  the {Gaussian} wiretap channel with a helper,'' in \emph{Proc. {IEEE} Int.
  Symp. Inf. Theory {(ISIT)}}, Jul. 2016, pp. 2978 -- 2982.

\bibitem{MGMK:14}
A.~S. Motahari, S.~O. Gharan, M.~A. {Maddah-Ali}, and A.~K. Khandani, ``Real
  interference alignment: Exploiting the potential of single antenna systems,''
  \emph{IEEE Trans. Inf. Theory}, vol.~60, no.~8, pp. 4799 -- 4810, Aug. 2014.

\end{thebibliography}


\end{document}